\documentclass[aap]{imsart}

\usepackage[utf8]{inputenc} 
\usepackage{amsmath, amssymb,bm, cases, amsthm, mathtools, thmtools}
\usepackage{verbatim}
\usepackage{graphicx}\graphicspath{{figures/}}
\usepackage{multicol}
\usepackage{tabularx}
\usepackage[usenames,dvipsnames]{xcolor}
\usepackage{mathrsfs} 
\usepackage{url}
\usepackage{bbm}

\usepackage[colorlinks,citecolor=blue,urlcolor=blue,linkcolor=RawSienna]{hyperref}
\usepackage[capitalize]{cleveref}

\usepackage{datetime}

\usepackage{euscript,microtype}

\crefname{lemma}{Lemma}{Lemmas}
\crefname{corollary}{Corollary}{Corollaries}
\crefname{theorem}{Theorem}{Theorems}

\makeatletter
\let\reftagform@=\tagform@
\def\tagform@#1{\maketag@@@{\ignorespaces\textcolor{gray}{(#1)}\unskip\@@italiccorr}}
\renewcommand{\eqref}[1]{\textup{\reftagform@{\ref{#1}}}}
\makeatother

\declaretheorem[style=plain,numberwithin=section,name=Theorem]{theorem}
\declaretheorem[style=plain,sibling=theorem,name=Lemma]{lemma}

\declaretheorem[style=plain,sibling=theorem,name=Corollary]{corollary}

\declaretheorem[style=definition,sibling=theorem,name=Definition]{definition}
\declaretheorem[style=definition,sibling=theorem,name=Example]{example}
\declaretheorem[style=remark,sibling=theorem,name=Remark]{remark}

\def\[#1\]{\begin{align}#1\end{align}}
\def\*[#1\]{\begin{align*}#1\end{align*}}

\def \bigO{\mathcal{O}}
\def \IG{\textbf{IG}}
\def \Normal{\mathcal{N}}

\newcommand{\Reals}{\mathbb{R}}

\newcommand{\Nats}{\mathbb{N}}

\newcommand{\dee}{\mathrm{d}}

\renewcommand{\Pr}{\mathbb{P}}
\def\EE{\mathbb{E}}
\newcommand{\var}{\textrm{var}}

\RequirePackage{amsthm,amsmath,amsfonts,amssymb}
\RequirePackage[numbers]{natbib}
\renewcommand{\citet}{\cite}

\startlocaldefs

\endlocaldefs

\begin{document}

\begin{frontmatter}
\title{Complexity Results for MCMC Derived from Quantitative Bounds}
\runtitle{Complexity Results for MCMC Derived from Quantitative Bounds}

\begin{aug}
\author[A]{\fnms{Jun}~\snm{Yang}
}
\and
\author[B]{\fnms{Jeffrey S.}~\snm{Rosenthal}
}
\address[A]{Department of Statistics,
	University of Oxford, UK
	}

\address[B]{Department of Statistical Sciences,
	University of Toronto, Canada
	}
\end{aug}

\begin{abstract}
This paper considers how to obtain MCMC quantitative convergence bounds which can be translated into tight complexity bounds in high-dimensional settings. We propose a modified drift-and-minorization approach, which establishes generalized drift conditions defined in subsets of the state space. The subsets are called the ``large sets'', and are chosen to rule out some ``bad'' states which have poor drift property when the dimension of the state space gets large.  Using the ``large sets'' together with a ``fitted family of drift functions'', a quantitative bound can be obtained which can be translated into a tight complexity bound.
As a demonstration, we analyze several Gibbs samplers and obtain complexity upper bounds for the mixing time.
In particular, for one example of Gibbs sampler which is related to the James--Stein estimator, we show that the number of iterations required for the Gibbs sampler to converge is constant under certain conditions on the observed data and the initial state. It is our hope that this modified drift-and-minorization approach can be employed in many other specific examples to obtain complexity bounds for high-dimensional Markov chains.

\end{abstract}



\end{frontmatter}

\section{Introduction}






Markov chain Monte Carlo (MCMC) algorithms are extremely widely used and studied in statistics, e.g. \cite{Brooks2011,Gilks1995}, and their running times are an extremely important practical issue. They have been studied from a variety of perspectives, including convergence ``diagnostics'' via the Markov chain output (e.g. \cite{Gelman1992}), proving weak convergence limits of speed-up versions of the algorithms to diffusion limits \cite{Roberts1997,Roberts1998}, and directly bounding the convergence in total variation distance
\cite{Meyn1994,Rosenthal1995a,Rosenthal1996,Roberts1999,Jones2001,Rosenthal2002,Flegal2008,baxendale2005renewal,jones2004sufficient}. Furthermore, there is a recent trend focusing on quantitative mixing time bounds in terms of either total variation distance or Wasserstein distance for certain types of MCMC methods (such as Langevin Monte Carlo) and targets (such as strongly log-concave targets), see e.g.~\cite{dalalyan2017theoretical,durmus2017nonasymptotic}. Among the work of directly bounding the total variation distance,
most of the quantitative convergence bounds proceed by establishing a \emph{drift condition} and an associated \emph{minorization condition} for the Markov chain in question (see e.g.~\cite{Meyn2012}).
One approach for finding quantitative bounds has been the drift and minorization method set forth by \citet{Rosenthal1995a}.

Computer scientists take a slightly different perspective, in terms of running time complexity order as the ``size'' of the problem goes to infinity. Complexity results in computer science go back at least to \citet{Cobham1965}, and took on greater focus with the pioneering NP-complete work of \citet{Cook1971}. In the Markov chain context, computer scientists have been bounding convergence times of Markov chain algorithms since at least \citet{Sinclair1989}, focusing largely on spectral gap bounds for Markov chains on finite state spaces. More recently, attention has turned to bounding spectral gaps of modern Markov chain algorithms on general state spaces, again primarily via spectral gaps, such as \cite{Lovasz2003,Vempala2005,Lovasz2006,Woodard2009,Woodard2009a} and the references therein. These bounds often focus on the order of the convergence time in terms of some particular parameter, such as the dimension of the corresponding state space. In recent years, there is much interest in the ``large $p$, large $n$'' or ``large $p$, small $n$'' high-dimensional {settings}, where $p$ is the number of parameters and $n$ is the sample size.  \citet{Rajaratnam2015} use the term convergence complexity to denote
the ability of a high-dimensional MCMC scheme to draw samples from the posterior, and how the ability to do so changes as the dimension of the parameter set grows.

Direct total variation bounds for MCMC are sometimes presented in terms of the convergence order, for example, the work by \citet{Rosenthal1995b} for a Gibbs sampler for a variance components model. However, current methods for obtaining total variation bounds of such MCMCs typically proceed as if the dimension of the parameter, $p$, and sample size, $n$, are fixed. It is thus important to bridge the gap between statistics-style convergence bounds, and computer-science-style complexity results.

In one direction,
\citet{Roberts2016} connect known results about diffusion limits of MCMC to the computer science notion of algorithm complexity. They show that any weak limit of a Markov process implies a corresponding complexity bound in an appropriate metric. For example, under appropriate assumptions, in $p$ dimensions, the Random-Walk Metropolis algorithm  takes $\bigO(p)$ iterations (see also \cite{yang2020optimal}) and the Metropolis-adjusted Langevin algorithm (MALA) takes $\bigO(p^{1/3})$ iterations to converge to stationarity.

This paper considers how to obtain MCMC quantitative convergence bounds that can be translated into tight complexity bounds in high-dimensional {settings}. At the first glance, it may seem that an approach to answering the question of convergence complexity may be provided by the drift-and-minorization method of \cite{Rosenthal1995a}. However, \citet{Rajaratnam2015} demonstrate that, somewhat problematically, a few specific upper bounds in the literature obtained by the drift-and-minorization method tend to $1$ as $n$ or $p$ tends to infinity. 
For example, by directly translating the existing work by \citet{Choi2013,Khare2013}, which are both based on the general approach of \cite{Rosenthal1995a}, \citet{Rajaratnam2015} show that the ``small set'' gets large fast as the dimension $p$ increases.
And this seems to happen generally when the drift-and-minorization approach is applied to statistical problems. \citet{Rajaratnam2015} also discuss special cases when the method of \cite{Rosenthal1995a} can still be used to obtain tight bounds on the convergence rate. However, the conditions proposed in \cite{Rajaratnam2015} are very restrictive.
First, it requires the MCMC algorithm to be analyzed is a Gibbs sampler. Second, the Gibbs sampler must have only one high-dimensional parameter which must be drawn in the last step of the Gibbs sampling cycle. 
Unfortunately, other than some tailored examples \cite{Rajaratnam2015}, most realistic MCMC algorithms do not satisfy these conditions. 
It is unclear whether some particular drift functions lead to bad complexity bounds or the drift-and-minorization approach itself has some limitations. It is therefore the hope by \citet{Rajaratnam2015} that proposals and developments of new ideas analogous to those of \cite{Rosenthal1995a}, which are suitable for high-dimensional settings, can be motivated.

In this paper, we attempt to address {concerns about obtaining} quantitative bounds that can be translated into tight complexity bounds. We note that although \citet{Rajaratnam2015} provide evidence for the claim that many published bounds have poor dependence on $n$ and $p$, the statistics literature has not focused on controlling the complexity order on $n$ and $p$. We give some {intuition} why most directly translated complexity bounds are quite loose and provide {advice} on how to obtain tight complexity bounds for high-dimensional Markov chains. The key ideas are (1) the drift function should be small near the region of concentration of the posterior in high dimensions; (2) ``bad'' states which have poor drift property when $n$ and/or $p$ gets large should be ruled out when establishing the drift condition. In order to get tight complexity bounds, we propose a modified drift-and-minorization approach by establishing generalized drift conditions in subsets of the state space, which are called the ``large sets'', instead of the whole state space; see \cref{section_drift}.  The ``large sets'' are chosen to rule out some ``bad'' states which have poor drift property when the dimension of the state space gets large. By establishing the generalized drift condition, a new quantitative bound is obtained, which is composed of two parts. The first part is an upper bound on the probability the Markov chain will visit the states outside of the ``large set''; the second part is an upper bound on the total variation distance of a constructed restricted Markov chain defined only on the ``large set''. In order to obtain good complexity bounds for high-dimensional {settings}, as the dimension increases, the family of drift functions should be chosen such that the function values are small near the region of concentration of the posterior, which we will define formally as a ``fitted family of drift functions'', and the ``large sets'' should be adjusted depending on $n$ and $p$ to balance the complexity order of the two parts. 

As a demonstration, we prove three Gibbs samplers to get complexity bounds. In the first two examples, we demonstrate how to choose the ``fitted family of drift functions''. In the third example, we demonstrate the use of ``fitted family of drift functions'' together with ``large sets''. More specifically, we show in \cref{subsec_james} that a certain realistic Gibbs sampler related to the James--Stein estimator converges in $\bigO(1)$ iterations; see \cref{coro_mixing_time}.
As far as we know, this is the first successful example for analyzing the convergence complexity of a \emph{non-trivial} realistic MCMC algorithm using the (modified) drift-and-minorization approach. Several months after we uploaded this manuscript to arXiv, \citet{Qin2017} successfully analyzed another realistic MCMC algorithm using the drift-and-minorization approach. Although the analysis by \citet{Qin2017} does not make use of the ``large set'' technique proposed in this paper, they do make use of a ``fitted family of drift functions'', which they use an informal concept called ``a centered drift function''.
We explain in this paper that when there exists some ``bad'' states, using a ``fitted family of drift functions'' might not be enough to establish a tight complexity bound. For example, for the Gibbs sampler we successfully analyze in \cref{subsec_james}, it is unknown how to obtain tight complexity bound by the traditional drift-and-minorization approach or other approaches. This is confirmed in a later study by \citet{davis2020convergence}.  To the best of our knowledge, our approach using the ``large set'' is the only successful approach so far to get the tight complexity bound of this example. For another successful example using the ``large set'', we refer to recent work in \cite{zhou2021} for high-dimensional Bayesian variable selection. An important message from the successful analysis of several MCMC examples using the ``large set'' together with a ``fitted family of drift functions'' is that complexity bounds can be obtained
even without any particular form of non-deteriorating convergence bounds. Previous attempts in the literature on studying how the geometric convergence rate behaves as a function of $p$ and $n$ are incomplete.
It is our hope that our approach can be employed to many other specific examples for obtaining quantitative bounds that can be translated to complexity bounds in high-dimensional settings.

\emph{{Notation}:} 
	We use  $\xrightarrow{d}$ for weak convergence and $\pi(\cdot)$ to denote the stationary distribution of the Markov chain. The total {variation} distance is denoted by $\|\cdot\|_{\var}$ and the law of a random variable $X$ denoted by $\mathcal{L}(X)$.
	We adopt the Big-O, Little-O, Theta, and Omega notations. Formally, $T(n)=\bigO(f(n))$ if and only if for some constants $c$ and $n_0$, $T(n)\le c f(n)$ for all $n\ge n_0$; $T(n)=\Omega(f(n))$ if and only if {for some constants $c$ and $n_0$,} $T(n)\ge c f(n)$ for all $n\ge n_0$; $T(n)$ is $\Theta(f(n))$ if and only if both $T(n)=\bigO(f(n))$ and $T(n)=\Omega(f(n))$; $T(n)=o(f(n))$ if and only if $T(n)=\bigO(f(n))$ and $T(n)\neq\Omega(f(n))$.

\section{Generalized Geometric Drift Conditions and Large Sets}\label{section_drift}

Scaling classical MCMCs to very high dimensions can be problematic. {Even if a chain is geometrically ergodic} for fixed $n$ and $p$, the convergence of Markov chains may still be quite slow as $p\to\infty$ and $n\to\infty$. 
Throughout the paper, we assume the Markov chain is positive Harris recurrent, aperiodic, and $\pi$-irreducible, where $\pi$ denotes the unique stationary distribution. For a Markov chain $\{X^{(i)},i=0,1,\dots\}$ on a state space $(\mathcal{X},\mathcal{B})$ with transition kernel $P(x,\cdot)$, defined by
\[
P(x,B)=\Pr(X^{(i+1)}\in B\,|\,X^{(i)}=x), \quad \forall x\in\mathcal{X}, B\in\mathcal{B}
\]
the general method of \cite{Rosenthal1995a} proceeds by establishing a \emph{drift condition}
\[\label{eq_old_drift}
\EE(f(X^{(1)})\,|\, X^{(0)}=x)\le \lambda f(x)+b, \quad \forall x\in \mathcal{X},
\]
where $f: \mathcal{X}\to \Reals^+$ is the ``drift function'', some $0<\lambda<1$ and $b<\infty$;
and an associated \emph{minorization condition}
\[
P(x,\cdot)\ge \epsilon Q(\cdot),\quad \forall x\in R,
\]
where $R:=\{x\in\mathcal{X}: f(x)\le d\}$ is called the ``small set'', and $d>2b/(1-\lambda)$, for some $\epsilon>0$ and some probability measure $Q(\cdot)$ on $\mathcal{X}$.
{Then \cite[Theorem 12]{Rosenthal1995a} states that under both drift and minorization conditions, if the Markov chain starts from an initial distribution $\nu$, then for any $0<r<1$, we have
\[\label{eq_rosenthal95bound}
\begin{split}
\|\mathcal{L}(X^{(k)})-\pi\|_{\var}
&\le(1-\epsilon)^{rk} +\alpha^{-k}(\alpha \Lambda)^{rk}\left[1+\EE_{\nu}(f(x))+\frac{b}{1-\lambda}\right],
\end{split}
\]
where $\alpha^{-1}=\frac{1+2b+\lambda d}{1+d}$, $\Lambda=1+2(\lambda d+b)$ and $\EE_{\nu}[f(x)]$ denotes the expectation of $f(x)$ over $x\sim \nu(\cdot)$.
}
However, it is observed, for example, in \cite{Rajaratnam2015,Qin2017}, that for many specific bounds obtained by the drift-and-minorization method, when the dimension gets larger, the typical scenario for the drift condition of \cref{eq_old_drift}  seems to be $\lambda$ going to one, and/or $b$ getting much larger. This makes the ``size'' of the small set $R$ {grow} too fast, which leads to the minorization volume $\epsilon$ {go} to $0$ exponentially fast. In the following, we give an intuitive explanation {of what} makes a ``good'' drift condition in high-dimensional {settings}.

\subsection{Intuition}

It is useful to think of the drift function $f(x)$ as an energy function
\cite{Jones2001}. Then the drift condition in \cref{eq_old_drift}
implies the chain tends to ``drift'' toward states which have
``lower energy'' in expectation. It is well-known that a ``good''
drift condition is established when both $\lambda$ and $b$ are
small. Intuitively, $\lambda$ being small implies that when the chain is
in a ``high-energy'' state, then it tends to ``drift''
back to ``low-energy'' states fast; and $b$ being small implies {that} when the chain
is in a ``low-energy'' state, then it tends
to remain in a ``low-energy'' state in the next iteration too.
In a high-dimensional setting as the dimension grows to infinity, for a
collection of drift conditions to be ``good'', we would like it to
satisfy the following two properties:

\medskip\noindent P1.
$\lambda$ is small, in the sense that it converges to $1$
slowly or is bounded away from $1$;

\medskip\noindent P2.
$b$ is small, in the sense that it grows at a slower rate
than do typical values of the drift function.

\medskip\noindent

We explain the intuition behind the properties and define a new notion of ``fitted family of drift functions'' in this subsection and later demonstrate how to establish the properties using examples in \cref{section_Gibbs}. One way to understand this intuition is to think of it as
controlling the complexity order of the size of the ``small set'',
$R=\{x\in\mathcal{X}: f(x)\le d\}$. Since $d>2b/(1-\lambda)$,
if $\lambda$ converges to $1$ slowly or is bounded away from $1$,
and if $b$ is growing at a slower rate than typical values of $f(x)$ ({we will illustrate the meaning of ``typical values'' later in examples}),
then the size of the small set parameter $d$ can be chosen to have a small
complexity order {on $n$ and/or $p$}.  This in turn makes the minorization volume $\epsilon$
converge to $0$ sufficiently slowly (or even remain bounded away from $0$).

Next, we define the notion of ``fitted family of drift functions'', which is somewhat related to the informal concept of ``centered drift function'' in \cite{Qin2017}.
\begin{definition}
    Let $\pi^p$ be the target distributions when the dimension of the state space is $p$. We call a collection of non-negative functions, $\{f_p(\cdot)\}_{p=1}^{\infty}$, a \emph{fitted family} if
    \[
    \lim_{p\to\infty}\EE_{\pi^{p}}[f_p(x)]=0.
    \]
\end{definition}
Then a fitted family of drift functions is just a fitted family of functions which also satisfy a family of (generalized) drift conditions. Note that the fitted family of functions can also depend on $n$ if $n$ is a function of $p$. In the rest of the paper, we may simply write $\pi^{p}$ as $\pi$ and $f_p(x)$ as $f(x)$ for simplicity. However, we should keep in mind that the notation $\pi$ and $f(x)$ are actually a family of target distributions and a fitted family of drift functions in high-dimensional settings when we study the behavior of the Markov chains for $p\to\infty$.

Now we explain the intuition on why we should use a fitted family of drift functions in
high-dimensional {settings}. For clarity, we first assume that $\lambda$ is bounded away from
$1$, and focus on conditions required for $b$ to grow at a slower
rate than typical values of $f(x)$. Assume for definiteness that
$p$ is fixed and $n\to\infty$, and the drift function is scaled in
such a way that $f(x)=\bigO(1)$ and there is a fixed typical state
$\tilde{x}$ with $f(\tilde{x})=\Theta(1)$ regardless of dimension.
Then, to satisfy property P2 above, we require that $b=o(1)$.  On the
other hand, taking expectation over $x\sim \pi(\cdot)$ on both sides
of \cref{eq_old_drift} yields $b\ge \EE_{\pi}[f(x)]/(1-\lambda)$,
so $b=\Omega(\EE_{\pi}[f(x)])$.  To make $b=o(1)$ implies that the
drift function should be chosen such that
$\EE_{\pi}[f(x)]\to 0$,
which is exactly the definition of the fitted family of drift functions.
Therefore, to get a small $b$ in a high-dimensional setting,
we require a (properly scaled) drift function $f(\cdot)$ whose
values $f(x)$, where
$x\sim\pi(\cdot)$, concentrate around $0$, which is guaranteed by the fitted family of drift functions.

Note that the fitted family of drift functions for high-dimensional settings can be very different than traditional ``good'' drift functions. For example, to study a Markov chain $\{X^{(k)}\}$ sampling a fixed-dimensional target $\pi$, one might think $f(x)=\pi(x)^{-\alpha}$ for some fixed number $\alpha>0$ is a good candidate for the drift function. However, this is not a good intuition for choosing the fitted family of drift functions for the high-dimensional settings. The following is a toy example.
\begin{example}
Consider $\pi$ is the standard multivariate Gaussian $\mathcal{N}(0,I_p)$.  One choice for the drift function could be $f(x)=\exp(\|x\|^2)-1$ or
$f(x)=\|x\|^2/p$ (which is similar to the one used in \cite[Example 1]{Rosenthal1995a}). However, a better fitted family of drift functions in high dimensional settings could be
\[
f(x)=(\|x\|^2/p-1)^2.
\]
This is because that under $X\sim \mathcal{N}(0,I_p)$, we know $\|X\|^2/p$ concentrates around $1$. The family of drift functions $\{(\|x\|^2/p-1)^2\}_{p=1}^{\infty}$ exactly fits this concentration phenomenon.  The traditional popular choices of drift functions do not have this property. 
\end{example}

Note that in the existing
literature, the drift functions used to establish the drift condition
usually don't satisfy the definition of fitted family of drift functions. This is because in the traditional
setting where $n$ and $p$ are fixed, a ``good'' drift condition is
established whenever $\lambda$ and $b$ are small enough for specific
fixed values of $n$ and $p$.  The complexity
orders of $\lambda$ and $b$ as functions of $n$ and/or $p$ are not
essential, so fitting the concentration region of the posterior as dimension increases is not
necessary. As a result, many
existing quantitative bounds cannot be directly translated into tight
complexity bounds, since the size of the small set does not have a small
complexity order {on $n$ and/or $p$}. At the very least, one has to re-analyze such MCMC
algorithms using a fitted family of drift functions. 

Next, we focus on establishing $\lambda$ that is either bounded away
from $1$ or converges to $1$ slowly, assuming a fitted family of drift functions is
already chosen.  Intuitively, $\lambda$ describes
the behavior of the Markov chain when its current state has a ``high
energy''. If $\lambda$ goes to $1$ very fast when $n$ and/or $p$ goes to
infinity, this may suggest the existence of some ``bad'' states, i.e.\
states which have ``high energy'', but the drift property
becomes poor as $n$ and/or $p$ gets large. Therefore, in high
dimensions, once the Markov
chain visits one of these ``bad'' states, it only slowly drifts
back toward to the corresponding small set. Since the
drift condition in \cref{eq_old_drift} 
must hold for all $x\in\mathcal{X}$,
the existence of ``bad'' states forces $\lambda$ go
to $1$ very fast.  And since the
small set is defined as $R=\{x\in\mathcal{X}: f(x)\le d\}$ where
$d>2b/(1-\lambda)$, the scenario $\lambda\to 1$ very fast forces $R$
to become very large,
and hence the minorization volume $\epsilon$ goes
to zero very fast. One perspective on this problem is that
the definition of drift condition in \cref{eq_old_drift}
is too restrictive, since it must hold for all states $x$, even the
bad ones.


In summary, we are able to establish a small $b$ as in P2 above by
using a fitted family of drift functions. However, the other
difficulty in establishing a small $\lambda$ as in P1 above is the
existence of some ``bad'' states when $n$ and/or $p$ gets large. Since
the traditional drift condition defined in \cref{eq_old_drift} is
restrictive, the traditional drift-and-minorization method is not
flexible enough to deal with these ``bad'' states.  In the following,
we instead propose a modified drift-and-minorization approach using
a generalized drift condition, where the drift function is defined
only in a ``large set''.  This allows us to rule out those ``bad''
states in high-dimensional cases.

\subsection{New Quantitative Bound}

We first relax the traditional drift condition and define a generalized drift condition which is established only on a subset of the state space.
Recall that $\{X^{(k)}\}$ denotes the Markov chain  on a state space $(\mathcal{X},\mathcal{B})$ with a transition kernel $P(x,\cdot), \forall x\in\mathcal{X}$. Let $P^k(x,\cdot)$ be the $k$-step transition kernel. Denote $R_0$ as the ``large set'', i.e., $R_0\in\mathcal{B}$ is a subset of $\mathcal{X}$.

\begin{definition}{(Generalized drift condition on a large set)}\label{new_drift_cond}
	There exists a drift function $f: \mathcal{X}\to \Reals^{+}$ such that for some $\lambda<1$ and $b<\infty$, 
	\[\label{trad_drift_cond_R0}
	\begin{split}
	& \EE(f(X^{(1)})\,|\, X^{(0)}=x)
	\le \lambda f(x)+b, \quad \forall x\in R_0,
	\end{split}
	\]
	and (C1) \emph{or} (C1') holds.
	
	\medskip\noindent (C1). The ``large set'' $R_0$ is defined by $R_0=\{x\in \mathcal{X}: f(x)\le d_0\}$ for some $d_0>0$.
	
	\medskip\noindent (C1'). The transition kernel $P(x,\cdot)$ is a composition of reversible {(with respect to $\pi$)} steps $P=\prod_{i=1}^I P_i$, i.e.\ , $P(x,\dee y)=\int_{(x_1,\dots,x_{I-1})\in\mathcal{X}\times\cdots\times \mathcal{X}} P_1(x,\dee x_1) P_2(x_1,\dee x_2)\cdots P_{I}(x_{I-1},\dee y)$, where $I\ge 1$ is a fixed integer, and
	\[\label{equation_C1prime}
	\EE(f(\tilde{X}^{(1)})\,|\, \tilde{X}^{(0)}=x)\le \EE(f(X^{(1)})\,|\, X^{(0)}=x),\quad\forall x \in R_0,
	\]
	where $\{\tilde{X}^{(k)}\}$ denotes a restricted Markov chain with a transition kernel  $\prod_{i=1}^I \tilde{P}_i$ where $\tilde{P}_i(x,\dee y):=P_i(x,\dee y)$ for $x,y\in R_0, x\neq y$, and $\tilde{P}_i(x,x):=1- P_i(x,R_0 \backslash \{x\}),\forall x\in R_0$.
\end{definition}

\begin{remark}
	Note that only one of (C1) and (C1') is required. For (C1'), the Markov chain needs to be either reversible or can be written as a composition of reversible steps. This condition is very mild since it is satisfied by most realistic MCMC algorithms. For example, full-dimensional and random-scan Metropolis-Hastings algorithms and random-scan Gibbs samplers are reversible, and their deterministic-scan versions can be written as a composition of reversible steps. For (C1), it is required that the ``large set'' is constructed using the drift function in a certain way but there is no restriction for the transition kernel $P$. If $R_0$ is constructed as in (C1) then 
	\cref{equation_C1prime} automatically holds. Therefore, one should verify (C1') if one hopes to have more flexibility for constructing $R_0$ than the particular way in (C1). Particularly, if the drift function $f(x)$ depends on all coordinates, it might be hard to control all the states in $\{x\in \mathcal{X}: f(x)\le d_0\}$ as the dimension increases. Then (C1') might be preferable. 
\end{remark}
\begin{remark}	
	To verify (C1') in \cref{new_drift_cond}, one has to check a new inequality $\EE(f(\tilde{X}^{(1)})\,|\, \tilde{X}^{(0)}=x)\le\EE(f(X^{(1)})\,|\, X^{(0)}=x)$. This inequality in (C1') implies the ``large set'' $R_0$ should be chosen such that the states in $R_0$ have ``lower energy'' on expectation. This is intuitive since we assume the ``bad'' states all have ``high energy'' and poor drift property when $n$ and/or $p$ gets large. One trick is to choose $R_0$ by ruling out some (but not too many) states with ``high energy'' even if the states are not ``bad''. In \cref{subsec_james}, we demonstrate the use of this trick to select the ``large set'' $R_0$ so that $\EE(f(\tilde{X}^{(1)})\,|\, \tilde{X}^{(0)}=x)\le\EE(f(X^{(1)})\,|\, X^{(0)}=x)$ can be easily verified. The constructed $R_0$ in \cref{subsec_james} satisfies (C1') but not (C1).
\end{remark}

Next, we propose a new quantitative bound, which is based on the generalized drift condition on a ``large set''.

\begin{theorem}\label{thm_new_drift}
	Suppose the Markov chain satisfies the generalized drift condition in \cref{new_drift_cond} on a ``large set'' $R_0$. Furthermore, for a ``small set'' $R:=\{x\in\mathcal{X}: f(x)\le d\}$ where $d>2b/(1-\lambda)$, the Markov chain also satisfies a minorization condition:
	\[\label{eq_minorization}
	P(x,\cdot)\ge \epsilon Q(\cdot), \quad \forall x\in R,
	\]
	for some $\epsilon>0$, some probability measure $Q(\cdot)$ on $\mathcal{X}$. Finally, suppose the Markov chain begins with an initial distribution $\nu$ such that $\nu(R_0)=1$. Then for any $0<r<1$, we have
	\[\label{new_quant_bound}
	\begin{split}
	\|\mathcal{L}(X^{(k)})-\pi\|_{\var}
	&\le(1-\epsilon Q(R_0))^{rk} +\frac{(\alpha \Lambda)^{rk}\left[1+\EE_{\nu}[f(x)]+\frac{b}{1-\lambda}\right]-\alpha^{rk}}{\alpha^k-\alpha^{rk}}\\
	&\quad +k\,\pi(R_0^c)+\sum_{i=1}^k \nu P^i (R_0^c),
	\end{split}
	\]
	where $\alpha^{-1}=\frac{1+2b+\lambda d}{1+d}$, $\Lambda=1+2(\lambda d+b)$, and $\nu P^i(\cdot):=\int_{\mathcal{X}} P^i(x,\cdot)\nu(\dee x)$. 
\end{theorem}
\begin{proof}
	See \cref{proof_thm_new_drift}.
\end{proof}	

\begin{remark}
Note that the new bound in \cref{thm_new_drift} assumes the Markov chain begins with an initial distribution $\nu$ such that $\nu(R_0)=1$. This assumption is not very restrictive since the ``large set'' ideally should include all ``good'' states. In high-dimensional {settings}, the Markov chain is not expected to converge fast beginning with any state (see \cref{discussion_initial_state} for discussions on initial states).  Furthermore, the use of ``warm start'' becomes popular recently, see e.g.~\cite{dwivedi2019log}. However, it doesn't directly relate to the large set. We only require that the initial distribution $\mu$ is supported in the large set. For example, $\mu$ can be a point mass. For the term $Q(R_0)$ in \cref{new_quant_bound}, it can be replaced by any lower bound of $Q(R_0)$. Since the ``large set'' is ideally chosen to include all ``good'' states, one can expect $Q(R_0)$ is at least bounded away from $0$. In particular, if we have established an upper bound for $P(x,R_0^c)$ with $x\in R$, then we can apply $\epsilon Q(R_0^c)\le P(x,R_0^c)$ to get an upper bound of $Q(R_0^c)$ which can be turned into a lower bound on $Q(R_0)$.
\end{remark}

\begin{remark}
In the proof of \cref{thm_new_drift}, the generalized drift condition in \cref{new_drift_cond} essentially implies a traditional drift condition in \cref{eq_old_drift} for a constructed ``restricted'' Markov chain only on the ``large set'' $R_0$. The first two terms in the upper bound \cref{new_quant_bound} are indeed an upper bound on the total variation distance of this constructed ``restricted'' Markov chain. Note that the general idea of studying the restriction of a Markov chain to some ``good'' subset of the state space has appeared in the literature, such as \cite{Martin2000,Dyer2003,Jerrum2004,Efthymiou2016,Mangoubi2017,Rudolf2018,Medina-Aguayo2019} and the references therein, in which different ways of restrictions have been considered for different reasons. {For example, \citet{Bou-Rabee2013} studied the rate of convergence of the MALA algorithm by a similar argument, which is later extended in \cite{Eberle2014} to study  contraction rate in Wasserstein distance w.r.t.~ Gaussian reference measure. However, the argument in \cite{Bou-Rabee2013} is only for the MALA algorithm and the proof technique is by constructing a restricted chain. Comparing with \cite{Bou-Rabee2013}, our \cref{thm_new_drift} is for general MCMC algorithms with weaker conditions in (C1) and (C1'). In the proof, we use either a trace chain or a restricted chain depending on which condition is satisfied. Most importantly, the motivation of this work is to obtain tight complexity bound which is quite different from \cite{Bou-Rabee2013}. In \cref{thm_new_drift}, the goal of considering a ``good'' subset of the state space is to obtain better control on the dependence on $n$ and $p$ for the upper bound.}
\end{remark}

\begin{remark}
The last two terms in the upper bound \cref{new_quant_bound} give an upper bound of the probability that the Markov chain will visit $R_0^c$ starting from either the initial distribution $\nu$ or the stationary distribution $\pi$. Therefore, the proposed method in \cref{thm_new_drift} is a generalized version of the classic drift-and-minorization method \cite{Rosenthal1995a} by allowing the drift condition to be established on a chosen ``large set''. {Indeed, if we choose $R_0=\mathcal{X}$, then \cref{new_quant_bound} is almost the same as \cref{eq_rosenthal95bound}, except slightly tighter due to the terms $\alpha^{rk}$.}
\end{remark}

\begin{remark}
One more note about \cref{new_quant_bound} is that the new bound does not decrease exponentially with $k$. For example, the term $k\,\pi(R_0^c)$ is linear increasing with $k$ for fixed $n$ and $p$. We emphasize that we do not aim to prove a Markov chain is geometrically ergodic here.  An upper bound which decreases exponentially with $k$ for fixed $n$ and $p$ does not guarantee to have a tight complexity order {on $n$ and/or $p$}, which has been discussed in \cite{Rajaratnam2015}. Instead, our new bound in \cref{new_quant_bound} is designed for controlling complexity orders of $n$ and/or $p$ for high-dimensional Markov chains. In \cref{subsec_james}, we obtain a tight complexity bound for a Gibbs sampler of a simple random effect model related to the James--Stein estimator. Previous unsuccessful attempts for the same Gibbs sampler (see \cite{davis2020convergence}), were focusing on how to obtain convergence bounds with geometric/polynomial rates as a function of $p$ and $n$. The successful analysis of the Gibbs sampler in the current paper implies that complexity bounds can be obtained even without any particular form of non-deteriorating convergence bounds.
\end{remark}

\subsection{Complexity Bound}

Note that mixing time is often defined uniformly over initial states, which is difficult to extend to general state spaces. In this paper, the term ``mixing time'' is defined depending on the initial state. The formal definition is given in the following.
{
\begin{definition}\label{def_mixingtime}
	For any $0<c<1$, {we define the mixing time} $K_{c,x}$ of a Markov chain $\{X^{(k)}\}$ with initial state $x$ by
	\[
	K_{c,x}:=\arg\min_{k}\left\{\|\mathcal{L}(X^{(k)})-\pi\|_{\var}\le c \right\}\quad \textrm{subject to } X^{(0)}=x.
	\]
\end{definition}
}

The proposed new bound in \cref{thm_new_drift} can be used to obtain complexity bounds in high-dimensional {settings}. The key is to balance the complexity orders of $k$ {on $n$ and/or $p$} required for both the first two terms and the last two terms of the upper bound in \cref{new_quant_bound} to be small. The complexity order of $k$ {on $n$ and/or $p$} for the first two terms to be small can be controlled by adjusting the ``large set''. The ``large set'' should be kept as large as possible provided that ``bad'' states have been ruled out.
{For the last two terms to be small, we should determine the growth rate of $k$ as a function of $n$ and $p$ so that}
\[
k\,\pi(R_0^c)+\sum_{i=1}^k \nu P^i (R_0^c)\to 0.
\]
This may involve (carefully) bounding the tail probability of the transition kernel, depending on the definition of the ``large set'' and the complexity order aimed to establish.

We give a direct corollary of \cref{thm_new_drift} on mixing time in terms of $p$. In general, mixing time in terms of both $n$ and $p$ can be obtained using \cref{thm_new_drift}.
\begin{corollary}\label{new_coro1}
Suppose \cref{thm_new_drift} has been established for every dimension $p$.
Let $\tilde{k}_p$ and $\hat{k}_p$ be sequences of positive integers as functions of $p$ such that both $\tilde{k}_p\to \infty$ and $\hat{k}_p\to\infty$ as $p\to\infty$. Furthermore, $\lim_{p\to\infty} \tilde{k}_p-\hat{k}_p\ge 0$ and
\[
&\tilde{k}_p\pi(R_0^c)+\sum_{i=1}^{\tilde{k}_p}\nu P^i(R_0^c)\to 0\\
&\frac{\log(2+\EE_{\nu}[f(x)]+\frac{b}{1-\lambda})}{-\log(1-\epsilon Q(R_0))}\frac{\log(\alpha\Lambda/(1-\epsilon Q(R_0))}{\log(\alpha)}=\bigO(\hat{k}_p).
\]
Then the mixing time of the MCMC starting from $\nu$ has the complexity order $\bigO(\hat{k}_p)$. 
\end{corollary}

Using \cref{new_coro1}, one can plug-in the orders of $b$, $1-\lambda$, and $\epsilon$ to get the complexity bound. The following result is directly from \cref{new_coro1}.
\begin{corollary}\label{new_coro2}
Suppose \cref{thm_new_drift} has been established for every dimension $p$ and $c_1,\cdots,c_5$ are non-negative constants such that $\frac{1}{\epsilon}=\bigO(p^{c_1})$, $\frac{1}{1-\lambda}=\bigO(p^{c_2})$, and $b=\bigO(p^{c_3})$. Also, $c_4>c_1$ and
\[
p^{c_4}\pi(R_0^c)+\sum_{i=1}^{p^{c_4}}\nu P^i(R_0^c)\to 0.
\]
Furthermore, if $Q(R_0^c)=o(1)$ and $\EE_{\nu}[f(x)]=\bigO(p^{c_5})$, then the mixing time starting from $\nu$ has the complexity order $\bigO(p^{c_1}\log(p^{c_5}+p^{c_2+c_3})\log(p^{c_2+c_3}))=\bigO(p^{c_1}(\log(p))^2)$.
\end{corollary}

We will discuss several MCMC examples in \cref{section_Gibbs} to demonstrate the use of the fitted family of drift functions and ``large sets'' to get complexity bounds.

\subsection{Discussions}

We finish this section by giving a few more remarks and discussions on our main results.
\begin{itemize}
    \item \emph{Geometric ergodicity:} The Markov chain to be analyzed  in \cref{thm_new_drift} does not have to be geometrically ergodic. The proof of \cref{new_quant_bound} only implies that, after ruling out ``bad'' states, a constructed ``restricted'' Markov chain defined on the ``large set'' is geometrically ergodic. Therefore, the new bound in \cref{new_quant_bound} can be used to analyze non-geometrically ergodic high-dimensional Markov chains. 
    \item \emph{Relation to spectral gaps:} Many approaches in MCMC literature bound the spectral gap of the corresponding Markov operator \cite{Lovasz2003,Vempala2005,Lovasz2006,Woodard2009,Woodard2009a}. However, on general state spaces, the spectral gap is zero for Markov chains which are not geometrically ergodic, even if they do converge to stationarity. Our results do not require the Markov chain to be geometrically ergodic. Instead, we only require the constructed ``restricted'' chain on the ``large set'' in our proof is geometrically ergodic. Therefore, we cannot connect our results to bounds on spectral gaps. Furthermore, we do not require the Markov chain to be reversible. So our results apply even in the non-reversible cases, which makes spectral gaps harder to study or interpret. For these reasons, we do not present the main results in terms of spectral gaps. 
    \item \emph{Other types of drift condition:} In this paper, we use the drift condition of the type in \cite{Rosenthal1995a}. There is another popular drift condition (e.g., in 
\cite{roberts2000rates}) and the connection between the two is well-known (see \cite[Lemma 3.1]{jones2004sufficient}). Therefore, it is straightforward to establish our main result using the other drift condition in \cite{roberts2000rates}. 
\item \emph{Complexity of MCMC estimators:} It would be nice to obtain rate of convergence (or non-asymptotic bounds) for general MCMC estimators. The proof techniques in the existing literature on establishing rate of convergence of MCMC estimators \cite{adamczak2015exponential,adamczak2008tail,paulin2015concentration,latuszynski2011rigorous,latuszynski2013nonasymptotic,rudolf2009explicit,rudolf2010error,rudolf2011explicit} requires certain conditions such as geometric/polynomial drift conditions, or spectral gaps. However, our result doesn't require establishing a geometric/polynomial drift condition or a spectral gap.  Therefore, it is not clear how to connect our complexity results to complexity of other MCMC estimators. This is certainly an interesting direction for future work.
\end{itemize}

\section{Gibbs Sampler Convergence Bound}\label{section_Gibbs}

In this section, we study several examples of Gibbs sampling to analyze the convergence complexity using the proposed approach. In \cref{subsec_gauss} and \cref{subsec_poisson}, we consider a simple Gaussian example and a hierarchical Poisson example. Simplified versions of both examples for fixed dimensions was originally studied in \cite[Example 1 and Example 2]{Rosenthal1995a} and the original mixing times have poor complexity orders in terms of dimensions. We study the extensions of them in the high-dimensional setting and obtain tight complexity bounds by choosing fitted families of drift functions. In \cref{subsec_james}, we study the MCMC model in \cite{Rosenthal1996} which is related to the James--Stein estimator. We demonstrate how to use both the fitted family of drift functions and the ``large sets'' to obtain a tight complexity bound.

Note that although the bound in \cref{thm_new_drift} contains different ``admissible" growth combinations such as
	of $b$, $1/(1-\lambda)$, and $1/\epsilon$ (see also \cref{new_coro2}), the minorization volume $\epsilon$ relies on the small set which is determined by both $b$ and $\lambda$. Furthermore, if $b$ is fairly large, it is not surprising that $\lambda$ can be bounded away from $1$. Therefore, we can summarize our general principle in analyzing all the three examples as follows.
\begin{enumerate}
    \item We first focus on choosing a fitted family of drift functions so that $\EE[f(x)]\le b$ where $b$ has a small order.
    \item Next, we establish the drift condition. If $\lambda$ from the drift condition goes to $1$ too fast, we apply the ``large set'' to rule out certain states. After the first two steps, we get a generalized drift condition which leads to a small set with reasonably ``size''. 
    \item Finally, we focus on establishing a (potentially multi-step) minorization condition to obtain $\epsilon$ which goes to zero slowly (or bounded away from $0$).
\end{enumerate}

\subsection{A Gaussian Toy Example}\label{subsec_gauss}
A bivariate Gaussian model was studied in \cite[Example 1]{Rosenthal1995a} as a demonstration of the drift-and-minorization approach. In this subsection, we study an extension of this example to the high-dimensional setting. Suppose our target $\pi$ is $\mathcal{N}(\mu,\Sigma)$, a $2p$-dimensional  multivariate Gaussian, where $\mu=\begin{pmatrix} \mu_1\\ \mu_2
\end{pmatrix}$ and $\Sigma = \begin{pmatrix}
\Sigma_{11} & \Sigma_{12} \\
\Sigma_{21} & \Sigma_{22} 
\end{pmatrix}$. 
To sample from the target distribution, we use a two-step Gibbs sampler as in \cite[Example 1]{Rosenthal1995a}.
Writing $X=\begin{pmatrix}X_1\\ X_2\end{pmatrix}$, the conditional distribution can be written as
\[
X_1\mid X_2=x_2 \sim \mathcal{N}\left(\mu_1 + \Sigma_{12}\Sigma_{22}^{-1}(x_2-\mu_2), \Sigma_{11}-\Sigma_{12}\Sigma_{22}^{-1}\Sigma_{21}\right)
\]
and similarly for $X_2\mid X_1$.

For simplicity, we only consider the setting such that $\mu_1=\mu_2=0$ and $\Sigma_{11}=\Sigma_{22}=I_d$ and $\Sigma_{12}=\Sigma_{21}=\frac{1}{2}I_d$. It is straightforward to extend our analysis to general cases of $\mu$ and $\Sigma$. The corresponding Gibbs sampler is 
\[
X_1^{(1)} &\sim \mathcal{N}\left(\frac{1}{2}X_2^{(0)},\frac{3}{4}I_p\right)\\
X_2^{(1)} &\sim \mathcal{N}\left(\frac{1}{2}X_1^{(1)},\frac{3}{4}I_p\right).
\]
Note that $X_1^{(0)}$ is not used in the updates.

If we choose a drift function similar to the one used in \cite{Rosenthal1995a}, such as
\[
f^{\textrm{old}}(X):=\|X_2\|^2/p.
\]
Then as $X_2^{(1)}\sim \mathcal{N}(\frac{1}{4}X_2^{(0)}, \frac{3}{4}(1+\frac{1}{4})I_p)$, it can be easily verified that the following drift condition can be established:
\[
\EE[f^{\textrm{old}}(X^{(1)})]\le \frac{1}{16} f^{\textrm{old}}(X^{(0)})+1.
\]
However, as $\|X_2\|^2/p$ concentrates to $1$ under stationarity, the drift condition leads to a small set $\{X: \|X_2\|^2/p=\bigO(1)\}$ in which the states that $\|X_2\|^2/p$ is much smaller than $1$ are included.

In our analysis, we choose a fitted family of drift functions which lead to a small set with much smaller size:
\[
f^{\textrm{new}}(X):=\left(\frac{\|X_2\|^2}{p}-1\right)^2.
\]
We can establish the following drift condition:
\[
\EE[f^{\textrm{new}}(X^{(1)})]\le \frac{1}{4} f^{\textrm{new}}(X^{(0)})+\bigO(1/p).
\]
The corresponding small set $\{X: 1-C/\sqrt{p}\le \|X_2\|^2/p\le 1+C/\sqrt{p}\}$ for some constant $C$, fits exactly the concentration region of the target as $p\to\infty$. Using the above drift condition and a multi-step minorization condition, we can obtain the mixing time is $\bigO(\log(p))$. Our main result is in the following.
\begin{theorem}\label{thm_toy_example}
For the two-step Gibbs sampler for our multivariate Gaussian model, suppose the initial state satisfies $\|X_2^{(0)}\|^2=\bigO(p)$, then there exists positive constants $C_1$, $C_2$ such that
\[
\|\mathcal{L}(X^{(n)})-\pi\|_{\var}\le \left[C_1+\left(\frac{\|X_2^{(0)}\|^2}{d}-1\right)^2\right]\gamma^k,
\]
where $\gamma<1$ is a fixed constant and number of steps $n=\lfloor kC_2\log(p)\rfloor+1$ where $k$ is any positive integer.
\end{theorem}
\begin{proof}
    See \cref{proof_thm_toy_example}.
\end{proof}
This implies the following complexity bound directly.
\begin{corollary}
Under the assumptions of \cref{thm_toy_example}, the mixing time of the Gibbs sampler is $\bigO(\log(p))$.
\end{corollary}

\subsection{A Hierarchical Poisson Model}\label{subsec_poisson}

We study a hierachical Poisson model originally for analyzing a realistic data set in \cite{gelfand1990sampling}. A Gibbs sampler for this model has been studied by \citet{gelfand1990sampling} and a (numerical) quantitative bound was studied using the drift-and-minorization approach in \cite[Example 2]{Rosenthal1995a}. In this subsection, we study the Gibbs sampler in the high-dimensional setting. 

Suppose the data has the form $\{Y_i, t_i\}_{i=1}^n$ where $Y_i$ represents the number of failures over a time interval $t_i$ of $n$ nuclear pumps. One can model the failures as a Poisson process with parameter $\lambda_i$. Thus, during a  observation period of length $t_i$, the number of failures $Y_i$ follows a Poisson distribution of parameters $\lambda_it_i$. We are interested in inferring the parameters $\lambda=(\lambda_1,\dots,\lambda_n)$ from the data $\{Y_i, t_i\}$. We follow a hierarchical Bayesian approach where we assume that  $\lambda_1,\dots,\lambda_n$ are conditional independent on a hyperparameter $\beta$ and follow a gamma distribution with density
\[
\pi(\lambda_i\mid \beta) =\frac{\beta^{\alpha-1}}{\Gamma(\alpha)}\lambda_i^{\alpha-1}\exp(-\beta\lambda_i)
\]
where $\alpha$ is a constant. We assume further that the hyperparameter $\beta$ follows itself a prior gamma distribution $\textrm{Ga}(\rho,\delta)$ where $\rho$ and $\delta$ are fixed constant
\[
\pi(\beta) =\frac{\delta^{\rho-1}}{\Gamma(\rho)}\beta^{\rho-1}\exp(-\delta\beta).
\]

For simplicity, in this example we assume the time intervals are unit, that is, $t_i=1$ for all $i$. 
It is straightforward to extend our analysis to general cases of time intervals.

Overall, the model can be written as
\[
Y_i\mid \lambda, \beta &\sim \textrm{Poisson}(\lambda_i),\quad i=1,\dots,n\\
\lambda_i\mid \beta &\sim_{\textrm{indep}}\textrm{Ga}(\alpha,\beta),\quad i=1,\dots,n,\\
\beta &\sim \textrm{Ga}(\rho,\delta).
\]
In this example, we have $p=n+1$ and $x=(\lambda_1,\dots,\lambda_n,\beta)$. The posterior satisfies
\[
\pi(x \mid Y_1,\dots,Y_n)\propto \pi(\beta)\pi(\lambda_i\mid \beta)\prod_{i=1}^n \frac{\lambda_i^{Y_i}}{Y_i!}\exp(-\lambda_i).
\]
Note that this multidimensional distribution is rather complicated and it is not obvious how the rejection sampling or importance sampling could be efficiently used in this context. As the conditional distributions $\pi(\lambda_1,\dots,\lambda_n \mid \beta, \{Y_i\})$ and $\pi(\beta\mid \{\lambda_i\}, \{Y_i\})$ admit standard parametric forms, we can write a Gibbs sampler with the following updating order:
\[
\pi(\lambda_1^{(k+1)},\dots,\lambda_n^{(k+1)}\mid \beta^{(k)}, \{Y_i\})&=\prod_{i=1}^n \pi(\lambda_i^{(k+1)}\mid \beta^{(k)}, Y_i)\\
\lambda_i^{(k+1)}\mid \beta^{(k)}, Y_i &\sim \textrm{Ga}(Y_i+\alpha, 1+\beta^{(k)}), \quad i=1,\dots,n\\
\beta^{(k+1)}\mid \{\lambda_i^{(k+1)}\}, \{Y_i\} &\sim \textrm{Ga}(\rho+ n\alpha, \delta+\sum_{i=1}^n \lambda_i^{(k+1)}).
\]

Next, we present the main result for this Gibbs sampler. The key step is to use a fitted family of drift functions:
\[
f_n(x):=\left(\frac{\sum_i \lambda_i}{n\alpha}-\frac{1}{\beta}\right)^2.
\]
Our main result for this Gibbs sampler is as follows.
\begin{theorem}\label{thm_new_example}
Suppose there exists a constant $N$ that for all $n\ge N$ the data satisfies $\bar{Y}:=\frac{1}{n}\sum_i Y_i\in [l,u]$ where $l$ and $u$ are two fixed constant such that $0<l<u<\infty$. Then there exists a constant $C$ such that for large enough $n$ and for all $k$, we have
\[
\|\mathcal{L}(X^{(k)})-\pi\|_{\var}\le \left[C+\left(\frac{\frac{1}{n}\sum_{i=1}^n\lambda_i^{(0)}}{\alpha}-\frac{1}{\beta^{(0)}}\right)^2\right]\gamma^k,
\]
where $\gamma<1$ is a constant.
\end{theorem}
\begin{proof}
    See \cref{proof_thm_new_example}.
\end{proof}

    Note that it is very reasonable to make some reasonable assumptions on the observed data since the posterior depends on the observed data and we are actually studying a sequence of posteriors for the convergence complexity. In \cref{thm_new_example}, we assume there exists a constant $N$ that for all $n\ge N$ the data satisfies $\bar{Y}:=\frac{1}{n}\sum_i Y_i\in [l,u]$ where $l$ and $u$ are two fixed constant such that $0<l<u<\infty$. This assumption is quite weak. For example, it holds if the data is indeed generated from the model with some ``true'' parameters.

\cref{thm_new_example} implies the following complexity bound directly.
\begin{corollary}
Under the assumptions of \cref{thm_new_example}, if the initial state satisfies $\frac{1}{n}\sum_{i=1}^n\lambda_i^{(0)}=\bigO(1)$ and $1/\beta^{(0)}=\bigO(1)$, the mixing time of the Gibbs sampler is $\bigO(1)$.
\end{corollary}

\subsection{A Random Effect Model related to the James--Stein Estimator}\label{subsec_james}

In this subsection, we concentrate on a particular MCMC model, which is related to the James--Stein estimator \cite{Rosenthal1996}:
	\[
	\begin{split}
	Y_i~|~\theta_i\quad &\sim \Normal(\theta_i,\sigma_V^2),\quad 1\le i\le n,\\
	\theta_i~|~\mu,{\sigma^2_A} &\sim \Normal(\mu,{\sigma^2_A}), \quad 1\le i\le n,\\
	\mu &\sim \textrm{ flat prior on }\mathbb{R},\\
	{\sigma^2_A} &\sim \IG(a,b),
	\end{split}
	\]
	where {${\sigma_V^2}$ is assumed to be known}, $(Y_1,\dots,Y_n)$ is the observed data, and $x=({\sigma^2_A},\mu, \theta_1,\dots,\theta_n)$ are parameters. Note that we have the number of parameters $p=n+2$ in this example. For simplicity, we will not mention $p$ but only refer to $n$ for this model. The posterior distribution satisfies
	\[
	\begin{split}
	&\mathcal{L}({\sigma^2_A},\mu,\theta_1,\dots,\theta_n~|~ Y_1,\dots,Y_n)\\
	&\propto \frac{b^a}{\Gamma(a)}({\sigma^2_A})^{-a-1}e^{-b/{\sigma^2_A}}\prod_{i=1}^n\frac{1}{\sqrt{2\pi {\sigma^2_A}}}e^{-\frac{(\theta_i-\mu)^2}{2{\sigma^2_A}}}\frac{1}{\sqrt{2\pi {\sigma_V^2}}}e^{-\frac{(Y_i-\theta_i)^2}{2{\sigma_V^2}}}.
	\end{split}
	\]
	
	A Gibbs sampler for the posterior distribution of this model has been originally analyzed in \cite{Rosenthal1996}. A quantitative bound has been derived by \citet{Rosenthal1996} using the drift-and-minorization method with a drift function $f(x)=\sum_{i=1}^n (\theta_i-\bar{Y})^2$ {where $\bar{Y}=\frac{1}{n}\sum_{i=1}^n Y_i$}. We first observe that this drift function doesn't lead to a fitted family of drift functions in high-dimensional setting. For example, select a ``typical'' state $\tilde{x}=({{\tilde{\sigma}^2_A}},\tilde{\mu},\tilde{\theta}_1,\dots,\tilde{\theta}_n)$ such that $\tilde{\theta}_i=Y_i$, we get $f(\tilde{x})=\sum_{i=1}^n(Y_i-\bar{Y})^2$. Under reasonable assumptions on the observed data $\{Y_i\}$, we can get the properly scaled drift function {$\frac{1}{n}f(\tilde{x})=\frac{1}{n}\sum_{i=1}^n(Y_i-\bar{Y})^2=\Theta(1)$}.
	Then {$b/n=\frac{1}{n}\sum_{i=1}^n(Y_i-\bar{Y})^2+\frac{n+1/4}{n}{\sigma_V^2}=\Theta(1)$} in \cite{Rosenthal1996}. Therefore, the definition of fitted family of drift functions is not satisfied. Furthermore, the established $\lambda$ in \cite{Rosenthal1996} converges to $1$ very fast, satisfying $1/(1-\lambda)=\Omega(n)$. Therefore, if we translate the quantitative bound in \cite{Rosenthal1996} into complexity orders, it requires the size of the ``small set'' to be $\Omega(n^2)$, which makes the minorization volume $\epsilon$ be exponentially small. This leads to upper bounds on the distance to stationarity which require exponentially large number of iterations to become small. This result also coincides with the observations by \citet{Rajaratnam2015} when translating the work of \citet{Khare2013,Choi2013}.
	
	\begin{remark}\label{key_remark}
	    When the dimension of the state space is fixed, \citet[Appendix C]{jones2006fixed} states a way to use rejection sampler to obtain samples from the posterior of this model. However, it is easy to show that the acceptance probability of the rejection sampler in \cite[Appendix C]{jones2006fixed} decreases very fast as the dimension increases. See \cref{proof_key_remark} for more details. As we will show the mixing time of our Gibbs sampler for this model is $\bigO(1)$, the rejection sampler in \cite[Appendix C]{jones2006fixed} is not as efficient in high dimensions as our Gibbs sampler. 
	\end{remark}

	We demonstrate the use of the modified drift-and-minorization approach by analyzing a Gibbs sampler for this MCMC model. Defining $x^{(k)}=(({\sigma^2_A})^{(k)},\mu^{(k)},\theta_1^{(k)},\dots,\theta_n^{(k)})$ to be the state of the Markov chain at the $k$-th iteration, we consider the following order of Gibbs sampling for computing the posterior distribution:
	\[
	\begin{split}
	\mu^{(k+1)}&\sim \mathcal{N}\left(\bar{\theta}^{(k)},\frac{({\sigma^2_A})^{(k)}}{n}\right),\\
	\theta_i^{(k+1)}&\sim \mathcal{N}\left(\frac{\mu^{(k+1)}{\sigma_V^2}+Y_i ({\sigma^2_A})^{(k)}}{{\sigma_V^2}+ ({\sigma^2_A})^{(k)}},\frac{({\sigma^2_A})^{(k)}{\sigma_V^2}}{{\sigma_V^2}+({\sigma^2_A})^{(k)}}\right),\quad i=1,\dots,n,\\
	({\sigma^2_A})^{(k+1)}&\sim \IG\left(a+\frac{n-1}{2},b+\frac{1}{2}\sum_{i=1}^n (\theta_i^{(k+1)}-\bar{\theta}^{(k+1)})^2\right).
	\end{split}
	\]
	Note that, in the language of \cite{jin2021convergence}, this is an ``out-of-order'' block Gibbs sampler, so inferences for the posterior distribution should
be based on a ``shifted'' output sample $((\sigma^2_A)^{(k)}, \mu^{(k+1)}, \{\theta_i^{(k+1)}\})$.  In any case, it still has the
same rate of convergence \cite[Proposition 3]{jin2021convergence} so our
convergence analysis applies to both our version and the original
block Gibbs version of \cite{Rosenthal1996}.
	
	We prove that convergence of the Gibbs sampler is actually very fast: the number of iterations required is $\bigO(1)$. More precisely, we first make the following assumptions on the observed data $\{Y_i\}$: there exists {$\delta>0$, ${\bar{\sigma}_V^2}<\infty$,} and a positive integer $N_0$, such that, {almost surely with respect to the randomness of $\{Y_i\}$}:
	{
	\[\label{eq_assumption}
	{\sigma_V^2}+\delta\le \frac{\sum_{i=1}^n (Y_i-\bar{Y})^2}{n-1}\le {\bar{\sigma}_V^2},\quad \forall n\ge N_0.
	\]
	}
	
		The {assumption} in \cref{eq_assumption} {is} quite natural. {For example, if the data is indeed generated from the model with a ``true'' variance $\sigma^2_A>0$ then \cref{eq_assumption} obviously holds}. {More generally, the upper bound is just to ensure $\sum_{i=1}^n (Y_i-\bar{Y})^2=\bigO(n)$.} For {the lower bound}, note that our MCMC model implies that the variance of $Y_i$ is larger than ${\sigma_V^2}$ because of the uncertainty of $\theta_i$. Actually, under the MCMC model, conditional on the parameter $\sigma^2_A$, the variance of the data $\{Y_i\}$ equals ${\sigma_V^2}+{\sigma^2_A}$. Therefore, the assumption in \cref{eq_assumption} is just to assume the observed data is not abnormal under the MCMC model when $n$ is large enough.    Note that only the existence of $\delta$ is required for establishing our main results. More precisely, the existence of $\delta$ is needed to obtain an upper bound for $\pi(R_0^c)$. If such  $\delta$ does not exist, the MCMC model is {(seriously)} misspecified so the posterior distribution of the parameter ${\sigma^2_A}$, which corresponds to the variance of a Normal distribution, may concentrate on $0$. In that case, our upper bound on $\pi(R_0^c)$ does not hold.

	Then we show that, under the assumption \cref{eq_assumption}, with initial state
	\[\label{initial_state}
	\bar{\theta}^{(0)}=\bar{Y},\quad 
	({\sigma^2_A})^{(0)}=\begin{cases}
	\frac{\sum_{i=1}^n (Y_i-\bar{Y})^2}{n-1}-{\sigma_V^2}, & \textrm{if }\frac{\sum_{i=1}^n (Y_i-\bar{Y})^2}{n-1}> {\sigma_V^2},\\
	\frac{\sum_{i=1}^n (Y_i-\bar{Y})^2}{n-1}, & \textrm{otherwise},
	\end{cases}
	\] 
	and $\mu^{(0)}$ arbitrary (since $\mu^{(0)}$ will be updated in the first step of the Gibbs sampler),
	the mixing time of the Gibbs sampler to guarantee small total variation distance to stationarity is bounded by some constant when $n$ is large enough.
\subsubsection{Main Results}
First, we obtain a quantitative bound for large enough $n$, which is given in the following theorem.
\begin{theorem}\label{thm_Gibbs}
	{Under the assumption \cref{eq_assumption},} with initial state \cref{initial_state}, there exists a positive integer $N$ which does not depend on $k$, some constants $C_1>0,C_2>0,C_3>0$ and $0<\gamma<1$, such that for all $n\ge N$ and for all $k$, we have
	\[\label{eq_bound_Gibbs}
	\|\mathcal{L}(X^{(k)})-\pi\|_{\var}\le C_1\gamma^k+C_2\frac{k(1+k)}{n}+C_3\frac{k}{\sqrt{n}}.
	\]
\end{theorem}
\begin{proof} 
	Let $\Delta=\sum_{i=1}^n (Y_i-\bar{Y})^2$ and $x=({\sigma^2_A},\mu,\theta_1,\dots,\theta_n)$. Define the fitted family of drift functions $\{f_n(x)\}$ by
	\[\label{eq_drift_function_Gibbs}
	f_n(x):=n(\bar{\theta}-\bar{Y})^2+n\left[\left(\frac{\Delta}{n-1}-{\sigma_V^2}\right)-{\sigma^2_A}\right]^2.
	\] 
	Let $x^{(k)}=(({\sigma^2_A})^{(k)},\mu^{(k)},\theta^{(k)}_1,\dots,\theta^{(k)}_n)$ be the state of the Markov chain at the $k$-th iteration, then we show in \cref{key_thm} (see \cref{proof_key_thm}) that
	\[\label{orig_drift_func}
	\EE[f_n(x^{(k+1)})\,|\,x^{(k)}]\le \left(\frac{({\sigma_V^2})^2+2{\sigma_V^2}({\sigma^2_A})^{(k)}}{({\sigma_V^2})^2+2{\sigma_V^2}({\sigma^2_A})^{(k)}+(({\sigma^2_A})^{(k)})^2}\right)^2 f_n(x^{(k)})+b,\quad\forall x^{(k)}\in\mathcal{X}
	\]
	where $b=\bigO(1)$.
	
Note that in \cref{orig_drift_func}, the term $\left(\frac{({\sigma_V^2})^2+2{\sigma_V^2}({\sigma^2_A})^{(k)}}{({\sigma_V^2})^2+2{\sigma_V^2}({\sigma^2_A})^{(k)}+(({\sigma^2_A})^{(k)})^2}\right)^2$ depends on the coordinate $A^{(k)}$ of the state $x^{(k)}$ and is not bounded away from $1$, since $({\sigma^2_A})^{(k)}$ can be arbitrarily close to $0$. Therefore, $\left(\frac{({\sigma_V^2})^2+2({\sigma_V^2})({\sigma^2_A})^{(k)}}{V^2+2{\sigma_V^2}({\sigma^2_A})^{(k)}+(({\sigma^2_A})^{(k)})^2}\right)^2$ cannot be bounded by some $\lambda$ such that $0<\lambda<1$ and we cannot directly establish the traditional drift condition \cref{eq_old_drift} by \cref{orig_drift_func}. In the following, we establish the generalized drift condition \cref{new_drift_cond} using a ``large set''.

According to \cref{eq_assumption}, for large enough $n$, we have $\frac{\Delta}{n-1}>{\sigma_V^2}$. Then, we choose a threshold $T$ such that, for large enough $n$, we have $0<T<\frac{\Delta}{n-1}-{\sigma_V^2}$. Defining
$\lambda_T:=\left(\frac{({\sigma_V^2})^2+2{\sigma_V^2}T}{({\sigma_V^2})^2+2{\sigma_V^2}T+T^2}\right)^2<1$,
we get
\[
\EE[f_n(x^{(k+1)})\,|\,x^{(k)}]\le \lambda_T\, f_n(x^{(k)})+b,\quad\forall x\in R_T.
\]
where the ``large set'', $R_T$, is defined by
\[\label{eq_large_set}
R_T:=\left\{x\in\mathcal{X}: \left[\left(\frac{\Delta}{n-1}-{\sigma_V^2}\right)-{\sigma^2_A}\right]^2\le \left[\left(\frac{\Delta}{n-1}-{\sigma_V^2}\right)-T\right]^2\right\}.
\]
In order to satisfy the new drift condition  in \cref{new_drift_cond}, we verify (C1'). Note that in our example the transition kernel of the Gibbs sampler can be written as a composition of reversible steps and only the last step of the Gibbs sampler updates the parameter ${\sigma^2_A}$ which is used for defining the ``large set'' $R_T$. Therefore, in order to verify \cref{equation_C1prime}, it suffices to check the last step if the value of the drift function increases by updating $x^{(k)}\in R_T$ to $x^{(k+1)}\in R_T^c$.
By the definition of $R_T$, we have
\[
\begin{split}
\left[\left(\frac{\Delta}{n-1}-{\sigma_V^2}\right)-({\sigma^2_A})^{(k)}\right]^2\le \left[\left(\frac{\Delta}{n-1}-{\sigma_V^2}\right)-T\right]^2,\quad \forall x^{(k)}\in R_T\\
\left[\left(\frac{\Delta}{n-1}-{\sigma_V^2}\right)-({\sigma^2_A})^{(k+1)}\right]^2>\left[\left(\frac{\Delta}{n-1}-{\sigma_V^2}\right)-T\right]^2,\quad \forall x^{(k+1)}\notin R_T.
\end{split}
\]
This implies the value of $f_n(x)$ increases if the Markov chain is outside of the ``large set'' after updating ${\sigma^2_A}$. Therefore, the generalized drift condition in \cref{new_drift_cond} is satisfied.

Now we can use \cref{thm_new_drift} to derive a quantitative bound for the Gibbs sampler. We first show in \cref{lemma_epsilon} (see \cref{proof_lemma_epsilon}) that
	if $T=\Theta(1)$, by choosing the size of the ``small set'' $R=\{x\in\mathcal{X}: f_n(x)\le d\}$ to satisfy $d=\bigO(1)$ and $d> \frac{b}{1-\lambda_T}$, {there exists a probability measure $Q(\cdot)$ such that} the Markov chain satisfies a minorization condition in \cref{eq_minorization} with the minorization {volumne} $\epsilon=\Theta(1)$. 
	
Next, we show in \cref{lemma_R0} (see \cref{proof_lemma_R0}) that with the initial state given by \cref{initial_state}, there exists a positive integer $N$, which does not depend on $k$, such that for all $n\ge N$, we have
	\[\label{eq_tail_bound}
	\begin{split}
	&k\,\pi(R_T^c)+ \sum_{i=1}^k P^i (x^{(0)},R_T^c)\\
	&\le   \frac{k}{\sqrt{n}} \frac{\sqrt{b}(2{\sigma_V^2}/\delta+1)}{\left|\left(\frac{\Delta}{n-1}-{\sigma_V^2}\right)-T\right|} + \frac{k(1+k)}{2n}\frac{b}{\left[\left(\frac{\Delta}{n-1}-{\sigma_V^2}\right)-T\right]^2}.
	\end{split}
	\]

Now we derive a quantitative bound for the Gibbs sampler for large enough $n$ by combing results together. 
First, from \cref{orig_drift_func}, we have $b=\bigO(1)$. Recall that $\lambda_T=\left(\frac{({\sigma_V^2})^2+2{\sigma_V^2}T}{({\sigma_V^2})^2+2{\sigma_V^2}T+T^2}\right)^2$. We obtain $\frac{b}{1-\lambda_T}=\bigO(1)$ by choosing $T=\Theta(1)$. Since $d>\frac{b}{1-\lambda_T}$, we can choose the size of small set to be $d=\bigO(1)$. Then we have shown that the minorization volume $\epsilon=\Theta(1)$. For $Q(R_T)$, we know that $P(x^{(0)},R_T^c)=\bigO(1/n)$, where $x^{(0)}\in R$. This implies that $\epsilon Q(R_T^c)=\bigO(1/n)$. Since $\epsilon=\Theta(1)$, we have $\epsilon Q(R_T)=\epsilon- \epsilon Q(R_T^c)=\Theta(1)$. Furthermore, by definition $\alpha^{-1}=\frac{1+2b+\lambda_T d}{1+d}<1$, it can be verified that $\alpha^{-1}$ is bounded away from $0$ when $T=\Theta(1)$ and $d=\bigO(1)$. Next, since $\Lambda=1+2(\lambda_T d+b)=\Theta(1)$, {ignoring the term $\alpha^{rk}$ in \cref{new_quant_bound},} we choose $r=\log(\alpha)/\log(\alpha\Lambda/(1-\epsilon Q(R_T)))$ to balance the order of $(1-\epsilon Q(R_T))^r$ and $\alpha^{-1}(\alpha\Lambda)^r$ and define $\gamma:=(1-\epsilon Q(R_T))^r=\alpha^{-1}(\alpha\Lambda)^r$. Then we have $\gamma=\Theta(1)$ and $0<\gamma<1$. Furthermore, since  $f_n(x^{(0)})=0$ for large enough $n$ and $\frac{b}{1-\lambda_T}=\bigO(1)$, we can pick a constant $C_1$ such that $C_1\ge 2+\frac{b}{1-\lambda_T}$ for large enough $n$.
Finally, we have $k\pi(R_T^c)+\sum_{i=1}^k P^i (x^{(0)},R_T^c)\le C_2\frac{k(1+k)}{n}+C_3\frac{k}{\sqrt{n}}$ by \cref{eq_tail_bound}, then \cref{thm_Gibbs} follows from \cref{thm_new_drift}.
\end{proof}

Next, we translate the quantitative bound in \cref{thm_Gibbs} into the convergence complexity in terms of mixing time using similar arguments as \cref{new_coro1} and \cref{new_coro2}. We show the convergence complexity is $\bigO(1)$. Intuitively, to make the term $C_1\gamma^k$  in \cref{eq_bound_Gibbs} arbitrarily small, $k$ needs to have a complexity order of $\bigO(1)$ since $\gamma$ does not depend on $n$. The residual terms $C_2\frac{k(1+k)}{n}+C_3\frac{k}{\sqrt{n}}\to 0$ when $k=o(\sqrt{n})$. Therefore, the complexity bound on the mixing time of the Gibbs sampler equals the smaller complexity order between $\bigO(1)$ and $o(\sqrt{n})$, which is $\bigO(1)$. The formal result is given in the following. 

\begin{theorem}\label{coro_mixing_time}
	For any $0<c<1$, recall the definition of the mixing time {$K_{c,x}$} in \cref{def_mixingtime}. We write {$K_{c,x}$} as {$K_{c,x}(n)$} to emphasize its dependence on $n$. Under the assumptions of \cref{thm_Gibbs}, {with initial state $x^{(0)}$ given by \cref{initial_state},} 
	there exists $N_c=\Theta(1)$ and $\bar{K}_c=\Theta(1)$ such that 
	\[
	\begin{split}
	{K_{c,x^{(0)}}(n)}&\le \bar{K}_c, \quad\forall n\ge N_c.
	\end{split}
	\]
\end{theorem}
\begin{proof}
See \cref{proof_coro_mixing_time}.
\end{proof}

\subsubsection{Initial state}\label{discussion_initial_state} The main results in \cref{thm_Gibbs} and \cref{coro_mixing_time} hold for a particular initial state given in \cref{initial_state}. We discuss other initial states than the one given in \cref{initial_state}. Note that the new bound in \cref{key_thm} holds for any initial state that is in the ``large set''. Therefore, we can extend the results in   \cref{thm_Gibbs} to get bounds when the Markov chain starts from some other initial states in the ``large set''. Recall the assumption on the observed data $\{Y_i\}$ in \cref{eq_assumption}, we have assumed there exists $\delta>0$ such that $\frac{\sum_{i=1}^n (Y_i-\bar{Y})^2}{n-1}\ge {\sigma_V^2}+\delta$ for large enough $n$. Note that the existence of such $\delta$ is sufficient to obtain the results in \cref{thm_Gibbs} and \cref{coro_mixing_time}. In order to get bounds when the MCMC algorithm starts from other initial states, we assume $\delta$ is known and establish upper bounds using $\delta$ explicitly. We define the ``large set'' \cref{eq_large_set} using $T=\delta$ and the extension of \cref{thm_Gibbs} is given in the following.
\begin{theorem}\label{thm_Gibbs_extension} Let $\Delta=\sum_{i=1}^n (Y_i-\bar{Y})^2$. {Under the assumption \cref{eq_assumption},} 
	if the Markov chain starts with any initial state $x^{(0)}\in R_{\delta}$ {(defined in \cref{eq_large_set} with $T=\delta$),}
	there exists a positive integer $N$, which does not depend on $k$, some constants $C_1>0,C_2>0,C_3>0,C_4>0$ and $0<\gamma<1$, such that for all $n\ge N$ and for all $k$, we have
	\[
	\|\mathcal{L}(X^{(k)})-\pi\|_{\var}\le [C_1+f_n(x^{(0)})]\gamma^k+C_2\frac{k(1+k)}{n}+C_3\frac{k}{\sqrt{n}}+C_4 f_n(x^{(0)})\frac{k}{n},
	\]
	where $f_n(\cdot)$ is the fitted family of drift functions defined in \cref{eq_drift_function_Gibbs}.
\end{theorem}
\begin{proof}
	Following the same proof of \cref{thm_Gibbs} by keeping the term $f_n(x^{(0)})$, the first two terms of the upper bound given in \cref{new_quant_bound} can be replaced by $[C_1+f_n(x^{(0)})]\gamma^k$ and the last term of the upper bound in \cref{new_quant_bound} can be replaced by $\sum_{i=1}^k P^i (x^{(0)},R_{\delta}^c)\le C_2\frac{k(1+k)}{n}+C_4 f_n(x^{(0)})\frac{k}{n}$.
\end{proof}
From \cref{thm_Gibbs_extension}, using similar arguments as \cref{new_coro1}, we can immediately obtain a complexity bound when the Markov chain starts within a subset of the ``large set'', which is given in the following. This result suggests that if the Markov chain starts from an initial state which is not ``too far'' from the state given in \cref{initial_state}, the Markov chain still mixes fast. The mixing time becomes $\bigO(\log n)$ instead of $\bigO(1)$.
\begin{corollary}
	{Under the assumption \cref{eq_assumption}, if the initial state of the Markov chain satisfies}  $x^{(0)}\in \left\{x\in R_{\delta}: f_n(x)=o(n/\log n) \right\}$, the mixing time of the Gibbs sampler {satisfies $K_{c,x^{(0)}}=\bigO(\log n)$ for any given $0<c<1$.}
\end{corollary}
Note that $\left\{x\in R_{\delta}: f_n(x)=o(n/\log n) \right\}$ defines a subset of the ``large set'' $R_{\delta}$, and the above result shows that the mixing time is $\bigO(\log n)$ if the initial state is in this subset. {The order $o(n/\log n)$ comes from a balance between $f_n(x^{(0)})\gamma^k$ and $f_n(x^{(0)})\frac{k}{n}$.} We conjecture the same complexity order of $\bigO(\log n)$ on the mixing time may hold even if the initial state is in a larger subset, for example $\left\{x^{(0)}\in R_{\delta}: f_n(x^{(0)})=\Theta(n) \right\}$. However, in order to prove this, we need to derive tighter upper bound of $\sum_{i=1}^k P^i (x^{(0)},R_{\delta}^c)$ which is a non-trivial task. We therefore leave it as an open problem.

Finally, we do not have upper bounds for the Markov chain when the initial state is outside of the ``large set'' since the new bound in \cref{thm_new_drift} requires the Markov chain starts within the ``large set''. For this particular Gibbs sampler example, numerical experiments suggest that, if the Markov chain starts from a ``bad'' state, the number of iterations required for the Markov chain to mix can be much larger than $\bigO(\log n)$. In high-dimensional {settings}, when the dimension of the state space goes to infinity, the Markov chain may not mix fast starting from any state. This observation is loosely consistent with various observations in \cite{Hairer2011}.

\subsubsection{Discussions}

We end this section by giving some further remarks and comments on the analysis of the Gibbs sampler.
\begin{itemize}
    \item \emph{Drift function:} In the proof of \cref{thm_Gibbs}, we actually used a fitted family of drift functions if we scale the drift functions in \cref{eq_drift_function_Gibbs} by $1/n$. To check this, we select a ``typical'' state $\tilde{x}=({{\tilde{\sigma}^2_A}},\tilde{\mu},\tilde{\theta}_1,\dots,\tilde{\theta}_n)$ such that $\tilde{\theta}_i=Y_i$ and ${\tilde{\sigma}^2_A}=\frac{\sum_{i=1}^n(Y_i-\bar{Y})^2}{n-1}$ then the scaled drift function $f_n(\tilde{x})/n=n({\sigma_V^2})^2/n=\Theta(1)$. We then hope to establish $b$ such that $b/n=o(1)$, or equivalently, $b=o(n)$. Indeed, the established generalized drift condition has $b=\bigO(1)=o(n)$, which implies the definition of fitted family of drift functions is satisfied for $\{f_n(x)/n\}$. 
    \item \emph{``Large set'':} The result in \cref{orig_drift_func} implies that those states whose value of ${\sigma^2_A}$ are close to zero are ``bad'' states. Therefore, the goal of choosing the ``large set'' in \cref{eq_large_set} is to {ruling out} those states. Note that we have applied the trick that ruling more states with ``high energy'' could make \cref{equation_C1prime} easier to establish. In the ``large set'' $R_T$ defined by \cref{eq_large_set}, we have also ruled out the states $x$ whose value of ${\sigma^2_A}$ are larger than $\left|\left(\frac{\Delta}{n-1}-{\sigma_V^2}\right)-T\right|+\left(\frac{\Delta}{n-1}-{\sigma_V^2}\right)$. Note that these states are not ``bad'' states. However, by ruling them out, {it is easy to establish \cref{equation_C1prime} as shown in the proof of \cref{thm_Gibbs}.}
    \item \emph{The upper bound in \cref{eq_tail_bound}:} Although the upper bound of $k\,\pi(R_T^c)+ \sum_{i=1}^k P^i (x^{(0)},R_T^c)$ shown in \cref{eq_tail_bound} is loose, it is already enough for showing the mixing time of the Gibbs sampler is $\bigO(1)$. The proof of \cref{lemma_R0} only makes use of the form of drift function and the definition of ``large set'', and does not depend on the particular form of the transition kernel of the Gibbs sampler. We expect that, in general, tighter upper bounds on $k\,\pi(R_T^c)+ \sum_{i=1}^k P^i (x^{(0)},R_T^c)$ could be obtained, depending on the choice of ``large set'' and the MCMC algorithm to be analyzed. This may involve carefully bounding the tail probability of the transition kernel.
    \item \emph{The constants in \cref{thm_Gibbs}:}
In \cref{thm_Gibbs}, we do not compute the constants $N$, $C_1$, $C_2$, and $C_3$ explicitly. Actually, $C_2$ is given explicitly in \cref{lemma_R0}. $C_3$ is given in \cref{lemma_R0} but it depends on the unknown constant $\delta>0$ from the assumption \cref{eq_assumption}. Furthermore, $C_1$ can be explicitly computed under much more tedious computations. Finally, $N$ depends on the unknown constant $N_0$ in \cref{eq_assumption} and the resulting concentration property of the posterior distribution for parameter ${\sigma^2_A}$ by \cref{eq_assumption}. Therefore, if we make stronger assumptions on the observed data $\{Y_i\}$, it is then possible to compute all the constants in \cref{thm_Gibbs} explicitly under tedious computations, though we do not pursue that here.
\end{itemize}


\begin{appendix}
	\section{Proof of Theorem \ref{thm_new_drift}}\label{proof_thm_new_drift}

Recall that $R$ denotes the ``small set'' and $R_0$ denotes the ``large set''. We first construct a  transition kernel for a ``restricted'' chain define on $R_0$, $\tilde{P}(x,\cdot), \forall x\in R_0$. One goal of this construction is that the stationary distribution of the kernel $\tilde{P}$ equals to the $\pi(\cdot)$ restricted on the ``large set'' $R_0$, i.e., $\pi'(\dee x):=\pi(\dee x)/\pi(R_0), \forall x\in R_0$. We consider two different constructions depending on (C1) or (C1') in \cref{new_drift_cond} holds.
\begin{itemize}
	\item If (C1) in \cref{new_drift_cond} holds, then we define the kernel $\tilde{P}$ as the transition kernel of the ``trace chain'' constructed as follows. Let $X^{(m)}$ be a Markov chain with kernel $P$, we define a sequence of random entrance time $\{m_i\}_{i\in\Nats}$ by $m_0:=\min\{m\ge 0: X^{(m)}\in R_0 \}$, $m_i:=\min\{m > m_{i-1}: X^{(m)}\in R_0 \}$. Then $\{X^{(m_i)}\}_{i\in \Nats}$ is the ``trace chain'' and the transition kernel $\tilde{P}(x,B):=\Pr(X^{(m_1)}\in B\,|\, X^{(m_0)}=x),\forall x\in R_0$. It is clear that the ``trace chain'' is obtained by ``stopping the clock'' when the original chain is outside $R_0$, the constructed $\tilde{P}$ is a valid transition kernel. It can be verified that the stationary distribution of this ``trace chain'' is $\pi'$.
	\item If (C1') in \cref{new_drift_cond} holds, then we construct the ``restricted chain'' using the kernel $\tilde{P}=\prod_{i=1}^I \tilde{P}_i$ where $\tilde{P}_i(x,\dee y):=P_i(x,\dee y)$ for $x,y\in R_0, x\neq y$, and $\tilde{P}_i(x,x):=1- P_i(x,R_0 \backslash \{x\}),\forall x\in R_0$. Note that since each $P_i$ is reversible, one can easily verify that each $\tilde{P}_i$ is also reversible and the stationary distribution of $\tilde{P}$ is $\pi'$.
\end{itemize}

Suppose that $X^{(m)}$ and $Y^{(m)}$ are two realizations of the Markov chain, where $X^{(m)}$ starts with the initial distribution $\nu(\cdot)$ and  $Y^{(m)}$ starts with the stationary distribution $\pi(\cdot)$. We define $\tilde{X}^{(m)}$ and $\tilde{Y}^{(m)}$ to be two realizations of a constructed ``restricted'' Markov chain on the ``large set'' with the transition kernel $\tilde{P}(x,\cdot), \forall x\in R_0$. 
We assume $\tilde{X}^{(m)}$ starts with the same initial distribution $\nu(\cdot)$ as $X^{(m)}$ and $\tilde{Y}^{(m)}$ starts with $\pi'(\cdot)$. {Since $\nu(R_0)=1$, we assume $X^{(0)}=\tilde{X}^{(0)}$.}
This rest of the proof is a modification of the original proof of the drift-and-minorization method using coupling in \cite{Rosenthal1995a}. 

We define the hitting times of $(\tilde{X}^{(m)},\tilde{Y}^{(m)})$ to $R\times R$ as follows.
\[
\begin{split}
t_1:&=\inf\{m\ge 0: (\tilde{X}^{(m)},\tilde{Y}^{(m)})\in R\times R\},\\
t_i:&=\inf\{m\ge t_{i-1}+1: (\tilde{X}^{(m)},\tilde{Y}^{(m)})\in R\times R\},\quad \forall i>1.
\end{split}
\]
Let $N_k:=\max\{i: t_i<k\}$. Then $N_k$ denotes the number of $(\tilde{X}^{(m)},\tilde{Y}^{(m)})$ to hit $R\times R$ in the first $k$ iterations. The following result gives an upper bound for $\|\mathcal{L}(X^{(k)})-\mathcal{L}(Y^{(k)})\|_{\var}$.
\begin{lemma}\label{lemma_tmp1} When the Markov chain satisfies the minorization condition in \cref{eq_minorization}, for any $j>0$, we have
	\[
	\begin{split}
	\|\mathcal{L}(X^{(k)})-\mathcal{L}(Y^{(k)})\|_{\var}\le& (1-\epsilon Q(R_0))^j+\Pr(N_k<j)\\
	&\quad +k\,\pi(R_0^c)+{\sum_{i=1}^k \nu P^i (R_0^c)}.
	\end{split}
	\]
\end{lemma}
\begin{proof}
	First, by triangle inequality
	\begin{equation}
	\begin{split}
	\|\mathcal{L}(X^{(k)})-\mathcal{L}(Y^{(k)})\|_{\var}&\le \|\mathcal{L}(\tilde{X}^{(k)})-\mathcal{L}(\tilde{Y}^{(k)})\|_{\var} + \|\mathcal{L}(X^{(k)})-\mathcal{L}(\tilde{X}^{(k)})\|_{\var}\\
	& \quad+ \|\mathcal{L}(Y^{(k)})-\mathcal{L}(\tilde{Y}^{(k)})\|_{\var}.
	\end{split}
	\end{equation}
	{By the coupling inequality $\|\mathcal{L}(X^{(k)})-\mathcal{L}(\tilde{X}^{(k)})\|_{\var}\le \Pr(X^{(k)}\neq \tilde{X}^{(k)})\le \sum_{m=1}^k\Pr(X^{(m)}\notin R_0)$, we have}
	\begin{equation}
	\begin{split}
	& \|\mathcal{L}(Y^{(k)})-\mathcal{L}(\tilde{Y}^{(k)})\|_{\var}+\|\mathcal{L}(X^{(k)})-\mathcal{L}(\tilde{X}^{(k)})\|_{\var}\\
	&\le \sum_{m=1}^k\Pr\left(Y^{(m)}\notin R_0\right)+\sum_{m=1}^k\Pr\left(X^{(m)}\notin R_0\right)\\
	&\le k\,\pi(R_0^c)+\sum_{i=1}^k {\nu P^i (R_0^c)}.
	\end{split}
	\end{equation}
	Finally, the Markov chain with kernel $\tilde{P}(x,\cdot)$ satisfies both drift condition
	\[\label{Ptilde_drift}
	\EE(f(\tilde{X}^{(1)})\,|\, \tilde{X}^{(0)}=x)\le \lambda f(x)+b, \quad \forall x\in R_0,
	\]
	and minorization condition
	\[
	\tilde{P}(x,\dee y)\ge [\epsilon Q(R_0)]\frac{Q(\dee y)}{Q(R_0)},\quad \forall x, y\in R_0.
	\]
	Using the result from \cite[Theorem 1]{Rosenthal1995a}, we have
	\[
	\|\mathcal{L}(\tilde{X}^{(k)})-\mathcal{L}(\tilde{Y}^{(k)})\|_{\var}\le (1-\epsilon Q(R_0))^j +\Pr(N_k<j).
	\]
\end{proof}

Next, we further upper bound the term $\Pr(N_k<j)$ slightly tighter than \cite{Rosenthal1995a}. Define the $i$-th gap of return times by
$r_i:=t_i-t_{i-1}, \forall i>1$, 
then 
\begin{lemma} \label{lemma_tmp2}
For any $\alpha>1$ and $j>0$, and $k>j$,
	\[
	\Pr(N_k<j)\le \frac{1}{\alpha^k-\alpha^j}\left[\EE\left(\prod_{i=1}^j \alpha^{r_i}\right)-\alpha^j\right].
	\]
\end{lemma}
\begin{proof}
	Note that $\{N_k<j\}=\{t_j\ge k\}=\{r_1+\dots+r_j\ge k\}$ and $r_1+\dots+r_j\ge j$ by definition.
	Then the result comes from Markov's inequality
	\[
	\begin{split}
	\Pr(N_k<j)&=\Pr(r_1+\dots+r_j\ge k)\\
	&=\Pr(\alpha^{r_1+\dots+r_j}-\alpha^j\ge\alpha^k-\alpha^j)\\
	&\le \frac{1}{\alpha^k-\alpha^j}\left[\EE\left(\prod_{i=1}^j \alpha^{r_i}\right)-\alpha^j\right].
	\end{split}
	\]
\end{proof}

Next, we bound  $\EE\left(\prod_{i=1}^j \alpha^{r_i}\right)$ following the exact same arguments as in \cite[Proof of Lemma 4 and Theorem 12]{Rosenthal1995a}, which gives
\[
\begin{split}
\EE\left(\prod_{i=1}^{j} \alpha^{r_i}\right)\le (\alpha \Lambda)^{j-1}\left[1+\EE_{\nu}(f(x))+\EE_{\pi'}(f(x))\right].
\end{split}
\]
By the drift condition for $\tilde{P}(x,\cdot)$ in \cref{Ptilde_drift},
taking expectations on both sides of \cref{Ptilde_drift} leads to
$\EE_{\pi'}(f(x))\le \frac{b}{1-\lambda}$. Therefore, setting $j=rk+1$ and combining all results together yields
\[
\begin{split}
\|\mathcal{L}(X^{(k)})-\pi\|_{\var}
&\le(1-\epsilon Q(R_0))^{rk+1} +\frac{(\alpha \Lambda)^{rk}\left[1+\EE_{\nu}(f(x))+\frac{b}{1-\lambda}\right]-\alpha^{rk+1}}{\alpha^k-\alpha^{rk+1}}\\
&\quad +k\,\pi(R_0^c)+\sum_{i=1}^k {\nu P^i (R_0^c)}.
\end{split}
\]
{Finally, we slightly relax the upper bound by replacing $\alpha^{rk+1}$ with $\alpha^{rk}$ in both the denominator and numerator. Then \cref{thm_new_drift} is proved by further relaxing $(1-\epsilon)^{rk+1}$ to $(1-\epsilon)^{rk}$.}

	\section{Proof of Theorem \ref{coro_mixing_time}}\label{proof_coro_mixing_time}
	
		Using \cref{thm_Gibbs}, one sufficient condition for 
	\[
	\|\mathcal{L}(X^{(k)})-\pi\|_{\var}\le c
	\] 
	is that $n\ge N$ and
\[
C_1\gamma^k\le \frac{c}{3}, \quad C_2\frac{(1+k)^2}{n}\le \frac{c}{3},\quad C_3\frac{k}{\sqrt{n}}\le \frac{c}{3}.
\]
This requires the number of iterations, $k$, satisfies
\[
\frac{\log (C_1)-\log (c/3)}{\log(1/\gamma)}\le k\le \min\left\{ \sqrt{\frac{c/3}{C_3}}\sqrt{n}-1,
\frac{c/3}{C_3}\sqrt{n}
\right\}.
\]
Note that any $k$ (if exists) satisfying the above equation provides an upper bound for the mixing time {$K_{c,x^{(0)}}(n)$}.

That is, for any $n\ge N$ such that
\[
\frac{\log (C_1)-\log (c/3)}{\log(1/\gamma)}\le
\min\left\{ \sqrt{\frac{c/3}{C_3}}\sqrt{n}-1,
\frac{c/3}{C_3}\sqrt{n}
\right\},
\]
which is equivalent to
\[
n\ge \max\left\{N, \left[\bar{K}_c\frac{3C_3}{c}\right]^2, \left[\left(\bar{K}_c+1\right)\sqrt{\frac{3C_3}{c}}\right]^2\right\}=:N_c,
\]
we have
$\bar{K}_c:= \frac{\log (C_1)-\log (c)+\log(3)}{\log(1/\gamma)}$ is an upper bound of the mixing time.

Finally, it can be seen that both $\bar{K}_c=\Theta(1)$ and $N_c=\Theta(1)$.

	\section{Proof of Lemma \ref{key_thm}}\label{proof_key_thm}
	
	\begin{lemma}\label{key_thm} Under the assumptions of \cref{thm_Gibbs},
	let $\Delta=\sum_{i=1}^n (Y_i-\bar{Y})^2$ and $x=({\sigma^2_A},\mu,\theta_1,\dots,\theta_n)$. Define the fitted family of drift functions $\{f_n(x)\}$ by
	\[\label{eq_drift_function_Gibbs2}
	f_n(x):=n(\bar{\theta}-\bar{Y})^2+n\left[\left(\frac{\Delta}{n-1}-{\sigma_V^2}\right)-{\sigma^2_A}\right]^2.
	\] 
	Let $x^{(k)}=(({\sigma^2_A})^{(k)},\mu^{(k)},\theta^{(k)}_1,\dots,\theta^{(k)}_n)$ be the state of the Markov chain at the $k$-th iteration, then we have
	\[\label{orig_drift_func2}
	\EE[f_n(x^{(k+1)})\,|\,x^{(k)}]\le \left(\frac{({\sigma_V^2})^2+2{\sigma_V^2}({\sigma^2_A})^{(k)}}{({\sigma_V^2})^2+2{\sigma_V^2}({\sigma^2_A})^{(k)}+(({\sigma^2_A})^{(k)})^2}\right)^2 f_n(x^{(k)})+b,\quad\forall x^{(k)}\in\mathcal{X}
	\]
	where $b=\bigO(1)$.
\end{lemma}

\begin{proof}
    
{In this proof, we write $f_n(x)$ as $f(x)$ for simplicity.} Recall that 
the order of Gibbs sampling for computing the first scan is:
\[
\begin{split}
\mu^{(1)}&\sim \mathcal{N}\left(\bar{\theta}^{(0)},\frac{({\sigma^2_A})^{(0)}}{n}\right),\\
\theta_i^{(1)}&\sim \mathcal{N}\left(\frac{\mu^{(1)}{\sigma_V^2}+Y_i ({\sigma^2_A})^{(0)}}{{\sigma_V^2}+ ({\sigma^2_A})^{(0)}},\frac{({\sigma^2_A})^{(0)}{\sigma_V^2}}{{\sigma_V^2}+({\sigma^2_A})^{(0)}}\right),\\
({\sigma^2_A})^{(1)}&\sim \IG\left(a+\frac{n-1}{2},b+\frac{1}{2}\sum_{i=1}^n (\theta_i^{(1)}-\bar{\theta}^{(1)})^2\right).
\end{split}
\]
It suffices to show that for
$\Delta=\sum_{i=1}^n (Y_i-\bar{Y})^2$ and \[
f(x)=n(\bar{\theta}-\bar{Y})^2+n\left[\left(\frac{\Delta}{n-1}-{\sigma_V^2}\right)-{\sigma^2_A}\right]^2,
\] 
we have
\[
\EE[f(x^{(1)})\,|\,x^{(0)}]\le \left(\frac{({\sigma_V^2})^2+2{\sigma_V^2}({\sigma^2_A})^{(0)}}{({\sigma_V^2})^2+2{\sigma_V^2}({\sigma^2_A})^{(0)}+(({\sigma^2_A})^{(0)})^2}\right)^2 f(x^{(0)})+b,
\]
where $b=\bigO(1)$.

Note that we can compute the expectation in $\EE[f(x^{(1)})\,|\,x^{(0)}]$ by three steps, according to the reverse order of the Gibbs sampling. To simplify the notation, we define $\sigma$-algebras that we condition on:
\[
\begin{split}
\mathcal{G}_{A}:&=\sigma(({\sigma^2_A})^{(0)},\{\theta_i^{(1)}\},\mu^{(1)}),\\
\mathcal{G}_{\theta}:&=\sigma(({\sigma^2_A})^{(0)},\{\theta_i^{(0)}\},\mu^{(1)}),\\
\mathcal{G}_{\mu}:&=\sigma(({\sigma^2_A})^{(0)},\{\theta_i^{(0)}\},\mu^{(0)}).
\end{split}
\]
Then we have
\[
\begin{split}
\EE[f(x^{(1)})\,|\,x^{(0)}]&=\EE[f(x^{(1)})\,|\,\mathcal{G}_{\mu}]
=\EE[\EE[\EE[f(x^{(1)})\,|\,\mathcal{G}_{A}]\,|\,\mathcal{G}_{\theta}]\,|\,\mathcal{G}_{\mu}].
\end{split}
\]
The three steps are as follows:
\begin{enumerate}
	\item Compute the expectation over $({\sigma^2_A})^{(1)}$ given $\{\theta_i^{(1)}\}$ and $\mu^{(1)}$. This is to compute the conditional expectation
	\[
	f'(x^{(1)}):=\EE[f(x^{(1)})\,|\, \mathcal{G}_A],
	\]
	where we write $\EE[\cdot\,|\,\mathcal{G}_A]$ to denote the the expectation is over (recall that $a$ and $b$ are constants from the prior $\IG(a,b)$)
	\[
	({\sigma^2_A})^{(1)}&\sim \IG\left(a+\frac{n-1}{2},b+\frac{1}{2}\sum_{i=1}^n (\theta_i^{(1)}-\bar{\theta}^{(1)})^2\right)
	\] 
	for given $\theta^{(1)}$ and $\mu^{(1)}$.
	\item Compute the expectation over $\{\theta_i^{(1)}\}$ given $\mu^{(1)}$. This is to compute the conditional expectation
	\[
	f''(x^{(1)}):=\EE[f'(x^{(1)})\,|\, \mathcal{G}_{\theta}],
	\]
	where we use $\EE[\cdot\,|\,\mathcal{G}_{\theta}]$ to denote the expectation is over
	\[
	\theta_i^{(1)}&\sim \mathcal{N}\left(\frac{\mu^{(1)}{\sigma_V^2}+Y_i ({\sigma^2_A})^{(0)}}{{\sigma_V^2}+ ({\sigma^2_A})^{(0)}},\frac{({\sigma^2_A})^{(0)}{\sigma_V^2}}{{\sigma_V^2}+({\sigma^2_A})^{(0)}}\right),\quad i=1,\dots,n,
	\]
	for given $\mu^{(1)}$ and $({\sigma^2_A})^{(0)}$.
	\item Compute the expectation over $\mu^{(1)}$. This is to compute the conditional expectation
	\[
	\EE[f(x^{(1)})\,|\, x^{(0)}]=\EE[f''(x^{(1)})\,|\, \mathcal{G}_{\mu}],
	\]
	where we have used $\EE[\cdot\,|\,\mathcal{G}_{\mu}]$ to denote the expectation is over
	\[
	\mu^{(1)}&\sim \mathcal{N}\left(\bar{\theta}^{(0)},\frac{({\sigma^2_A})^{(0)}}{n}\right)
	\]
	for given $\{\theta_i^{(0)}\}$ and $({\sigma^2_A})^{(0)}$.
\end{enumerate}
In the following, we compute the three steps, respectively. We use $\bigO(1)$ to denote terms that can be upper bounded by some constant that does not depend on the state.
\subsection{Compute $f'(x^{(1)})=\EE[f(x^{(1)})\,|\, \mathcal{G}_{A}]$}

{The first term of $f(x^{(1)})$ is $n(\bar{\theta}^{(1)}-\bar{Y})^2$, which is $\mathcal{G}_{A}$-measurable by construction. Thus, $\EE[n(\bar{\theta}^{(1)}-\bar{Y})^2\,|\,\mathcal{G}_{A}]=n(\bar{\theta}^{(1)}-\bar{Y})^2$. Then}
\[
\begin{split}
f'(x^{(1)})&=\EE[f(x^{(1)})\,|\, \mathcal{G}_{A}]\\
&=n(\bar{\theta}^{(1)}-\bar{Y})^2+n\EE\left\{\left[\left(\frac{\Delta}{n-1}-{\sigma_V^2}\right)-({\sigma^2_A})^{(1)}\right]^2\,|\,\mathcal{G}_A\right\}.
\end{split}
\]
Note that
\[
\begin{split}
&n\EE\left\{\left[\left(\frac{\Delta}{n-1}-{\sigma_V^2}\right)-({\sigma^2_A})^{(1)}\right]^2\,|\,\mathcal{G}_A\right\}\\
&=n\left(\frac{\Delta}{n-1}-{\sigma_V^2}\right)^2 +n \EE[(({\sigma^2_A})^{(1)})^2\,|\,\mathcal{G}_A]-2n\left(\frac{\Delta}{n-1}-{\sigma_V^2}\right)\EE[({\sigma^2_A})^{(1)}\,|\,\mathcal{G}_A].
\end{split}
\]
Recall that $\EE[\cdot\,|\,\mathcal{G}_A]$ denotes that the expectation is over
\[
({\sigma^2_A})^{(1)}&\sim \IG\left(a+\frac{n-1}{2},b+\frac{1}{2}\sum_{i=1}^n (\theta_i^{(1)}-\bar{\theta}^{(1)})^2\right),
\] 
where $a$ and $b$ are constants from the prior $\IG(a,b)$. The mean and variance of $({\sigma^2_A})^{(1)}$ can be written in closed forms since $({\sigma^2_A})^{(1)}$ follows from an inverse Gamma distribution. Denoting $S:=\frac{\sum_i(\theta_i^{(1)}-\bar{\theta}^{(1)})^2}{n-1}$, we can write the mean of $({\sigma^2_A})^{(1)}$ using $S$ as follows:
\[
\begin{split}
\EE[({\sigma^2_A})^{(1)}\,|\,\mathcal{G}_A]&=\frac{\sum_i(\theta_i^{(1)}-\bar{\theta}^{(1)})^2+2b}{n-1+2(a-1)}\\
&=\frac{\sum_i(\theta_i^{(1)}-\bar{\theta}^{(1)})^2}{n-1}+\frac{2b}{n-1+2(a-1)}\\
&\quad-\left(\frac{\sum_i(\theta_i^{(1)}-\bar{\theta}^{(1)})^2}{n-1}\right)\left(\frac{2(a-1)}{n-1+2(a-1)}\right)\\
&=S+\bigO(1/n)+\bigO(1/n)S.
\end{split}
\]
Similarly, the variance of $({\sigma^2_A})^{(1)}$ can be written in terms of $S$ as well:
\[
\begin{split}
\var[({\sigma^2_A})^{(1)}\,|\,\mathcal{G}_A]&=\frac{(\sum_i(\theta_i^{(1)}-\bar{\theta}^{(1)})^2/2+b)^2}{[(n-1)/2+(a-1)]^2[(n-1)/2+(a-2)]}\\
&=\frac{1}{(n-1)/2+(a-2)}\left(\EE[({\sigma^2_A})^{(1)}\,|\,\mathcal{G}_A]\right)^2\\
&=\bigO(1/n)\left(S+\bigO(1/n)+\bigO(1/n)S\right)^2\\
&=\bigO(1/n)S^2+\bigO(1/n^2)S+\bigO(1/n^3).
\end{split}
\]
Substituting the mean and variance of $({\sigma^2_A})^{(1)}$ in terms of $S$, we have
\[
\begin{split}
f'(x^{(1)})&=\EE[f(x^{(1)})\,|\,\mathcal{G}_A]\\
&=n(\bar{\theta}^{(1)}-\bar{Y})^2+n\left(\frac{\Delta}{n-1}-{\sigma_V^2}\right)^2 +n S^2-2n\left(\frac{\Delta}{n-1}-{\sigma_V^2}\right)S\\
&\quad+\bigO(1)+\bigO(1)S+\bigO(1)S^2.
\end{split}
\]
\subsection{Compute $	f''(x^{(1)})=\EE[f'(x^{(1)})\,|\, \mathcal{G}_{\theta}]$}

Note that the terms in $f'(x^{(1)})$ involving $\{\theta_i^{(1)}\}$ are $(\bar{\theta}^{(1)}-\bar{Y})^2$ and $S=\frac{\sum_i(\theta_i^{(1)}-\bar{\theta}^{(1)})^2}{n-1}$. Then 
\[
\begin{split}
f''(x^{(1)})&=\EE[f'(x^{(1)})\,|\, \mathcal{G}_{\theta}]\\
&=n\EE\left[(\bar{\theta}^{(1)}-\bar{Y})^2\,|\,\mathcal{G}_{\theta}\right]+n\left(\frac{\Delta}{n-1}-{\sigma_V^2}\right)^2 \\
&\quad+n \EE[S^2\,|\,\mathcal{G}_{\theta}]-2n\left(\frac{\Delta}{n-1}-{\sigma_V^2}\right)\EE[S\,|\,\mathcal{G}_{\theta}]\\
&\quad+\bigO(1)+\bigO(1)\EE[S\,|\,\mathcal{G}_{\theta}]+\bigO(1)\EE[S^2\,|\,\mathcal{G}_{\theta}].
\end{split}
\]
Therefore, it suffices to compute the following terms
\[
\EE\left[ (\bar{\theta}^{(1)}-\bar{Y})^2\,|\,\mathcal{G}_{\theta}\right],\quad \EE[S\,|\,\mathcal{G}_{\theta}],\quad \EE[S^2\,|\,\mathcal{G}_{\theta}].
\]

Note that $\{\theta_i^{(1)}\}$ are independent (but not identically distributed) conditional on $\mathcal{G}_{\theta}$. For the first term $\EE\left[ (\bar{\theta}^{(1)}-\bar{Y})^2\,|\,\mathcal{G}_{\theta}\right]$, we have
\[
\begin{split}
\EE\left[ (\bar{\theta}^{(1)}-\bar{Y})^2\,|\,\mathcal{G}_{\theta}\right]&=\EE\left[ \left(\bar{\theta}^{(1)}-\frac{\mu^{(1)}{\sigma_V^2}+\bar{Y} ({\sigma^2_A})^{(0)}}{{\sigma_V^2}+ ({\sigma^2_A})^{(0)}}+\frac{\mu^{(1)}{\sigma_V^2}+\bar{Y} ({\sigma^2_A})^{(0)}}{{\sigma_V^2}+ ({\sigma^2_A})^{(0)}}-\bar{Y}\right)^2\,|\,\mathcal{G}_{\theta}\right]\\
&=\EE\left[ \left(\bar{\theta}^{(1)}-\frac{\mu^{(1)}{\sigma_V^2}+\bar{Y} ({\sigma^2_A})^{(0)}}{{\sigma_V^2}+ ({\sigma^2_A})^{(0)}}\right)^2\,|\,\mathcal{G}_{\theta}\right]+\left(\frac{\mu^{(1)}{\sigma_V^2}+\bar{Y} ({\sigma^2_A})^{(0)}}{{\sigma_V^2}+ ({\sigma^2_A})^{(0)}}-\bar{Y}\right)^2\\
&\quad+2\left(\frac{\mu^{(1)}{\sigma_V^2}+\bar{Y} ({\sigma^2_A})^{(0)}}{{\sigma_V^2}+ ({\sigma^2_A})^{(0)}}-\bar{Y}\right)\EE\left[ \left(\bar{\theta}^{(1)}-\frac{\mu^{(1)}{\sigma_V^2}+\bar{Y} ({\sigma^2_A})^{(0)}}{{\sigma_V^2}+ ({\sigma^2_A})^{(0)}}\right)\,|\,\mathcal{G}_{\theta}\right]\\
&=\var[\bar{\theta}^{(1)}\,|\,\mathcal{G}_{\theta}]+\left(\frac{{\sigma_V^2}}{{\sigma_V^2}+ ({\sigma^2_A})^{(0)}}\right)^2\left(\mu^{(1)}-\bar{Y}\right)^2\\
&=\frac{1}{n}\frac{({\sigma^2_A})^{(0)}{\sigma_V^2}}{{\sigma_V^2}+({\sigma^2_A})^{(0)}}+\left(\frac{{\sigma_V^2}}{{\sigma_V^2}+ ({\sigma^2_A})^{(0)}}\right)^2\left(\mu^{(1)}-\bar{Y}\right)^2
\end{split}
\]
For the other two terms involving $S$, we have the following lemma.

\begin{lemma}\label{temp_lemma} For $S=\frac{\sum_i(\theta_i^{(1)}-\bar{\theta}^{(1)})^2}{n-1}$, we have
	\[
	\EE[S\,|\,\mathcal{G}_{\theta}]=\frac{({\sigma^2_A})^{(0)}{\sigma_V^2}}{{\sigma_V^2}+({\sigma^2_A})^{(0)}}+\left(\frac{({\sigma^2_A})^{(0)}}{{\sigma_V^2}+({\sigma^2_A})^{(0)}}\right)^2\frac{\Delta}{n-1},\quad
	\var[S\,|\,\mathcal{G}_{\theta}]=\bigO(1/n).
	\]
\end{lemma}
\begin{proof}
	Define $\eta_i:=\theta_i^{(1)}-\frac{Y_i ({\sigma^2_A})^{(0)}}{{\sigma_V^2}+({\sigma^2_A})^{(0)}}$ then $\bar{\eta}=\bar{\theta}^{(1)}-\frac{\bar{Y} ({\sigma^2_A})^{(0)}}{{\sigma_V^2}+({\sigma^2_A})^{(0)}}$. Note that $\{\eta_i\}$ are i.i.d. conditional on $\mathcal{G}_{\theta}$ with
	\[
	\eta_i\sim \mathcal{N}\left(\frac{\mu^{(1)}{\sigma_V^2}}{{\sigma_V^2}+({\sigma^2_A})^{(0)}},\frac{({\sigma^2_A})^{(0)}{\sigma_V^2}}{{\sigma_V^2}+({\sigma^2_A})^{(0)}}\right),\quad 	\bar{\eta}\sim \mathcal{N}\left(\frac{\mu^{(1)}{\sigma_V^2}}{{\sigma_V^2}+({\sigma^2_A})^{(0)}},\frac{1}{n}\frac{({\sigma^2_A})^{(0)}{\sigma_V^2}}{{\sigma_V^2}+({\sigma^2_A})^{(0)}}\right).
	\]
	Next, we decompose $\sum_{i=1}^n(\theta_i^{(1)}-\bar{\theta}^{(1)})^2$ by
	\[\label{tmp_eq2}
	\begin{split}
	\sum_{i=1}^n(\theta_i^{(1)}-\bar{\theta}^{(1)})^2&=\sum_{i=1}^n\left(\eta_i+\frac{Y_i ({\sigma^2_A})^{(0)}}{{\sigma_V^2}+({\sigma^2_A})^{(0)}}-\bar{\eta}-\frac{\bar{Y} ({\sigma^2_A})^{(0)}}{{\sigma_V^2}+({\sigma^2_A})^{(0)}}\right)^2\\
	&=\sum_{i=1}^n\left((\eta_i-\bar{\eta})^2+\left(\frac{({\sigma^2_A})^{(0)}}{{\sigma_V^2}+({\sigma^2_A})^{(0)}}\right)^2(Y_i-\bar{Y})^2+\frac{2(\eta_i-\bar{\eta})(Y_i-\bar{Y})({\sigma^2_A})^{(0)}}{{\sigma_V^2}+({\sigma^2_A})^{(0)}}\right).
	\end{split}
	\]
	Then we can obtain $\EE[S\,|\,\mathcal{G}_{\theta}]$ by
	\[
	\begin{split}
		\EE[S\,|\,\mathcal{G}_{\theta}]&=\EE\left\{\left[\frac{\sum_i(\theta_i^{(1)}-\bar{\theta}^{(1)})^2}{n-1}\right]\,|\,\mathcal{G}_{\theta}\right\}\\
		&=\EE\left\{\left[\frac{\sum_i(\eta_i-\bar{\eta})^2}{n-1}\right]\,|\,\mathcal{G}_{\theta}\right\}+\left(\frac{({\sigma^2_A})^{(0)}}{{\sigma_V^2}+({\sigma^2_A})^{(0)}}\right)^2\frac{\sum_{i=1}^n(Y_i-\bar{Y})^2}{n-1}\\
		&=\frac{({\sigma^2_A})^{(0)}{\sigma_V^2}}{{\sigma_V^2}+({\sigma^2_A})^{(0)}}+\left(\frac{({\sigma^2_A})^{(0)}}{{\sigma_V^2}+({\sigma^2_A})^{(0)}}\right)^2\frac{\Delta}{n-1}.
	\end{split}
	\]

	For $\var[S\,|\,\mathcal{G}_{\theta}]$, using the Cauchy-Schwartz inequality
	\[
	\begin{split}
	&\var[S\,|\,\mathcal{G}_{\theta}]=\EE\left[(S-\EE[S\,|\,\mathcal{G}_{\theta}])^2\,|\,\mathcal{G}_{\theta}\right]\\
	&=\EE\left[\left(\frac{\sum_{i=1}^n(\eta_i-\bar{\eta})^2}{n-1}-\EE_{\{\eta_i\}}\left[\frac{\sum_i(\eta_i-\bar{\eta})^2}{n-1}\right]+2\frac{({\sigma^2_A})^{(0)}}{{\sigma_V^2}+({\sigma^2_A})^{(0)}}\frac{\sum_{i=1}^n(\eta_i-\bar{\eta})(Y_i-\bar{Y})}{n-1}\right)^2\,|\,\mathcal{G}_{\theta}\right]\\
	&\le 2\var\left[\frac{\sum_i(\eta_i-\bar{\eta})^2}{n-1}\,|\,\mathcal{G}_{\theta}\right] +8\left(\frac{({\sigma^2_A})^{(0)}}{{\sigma_V^2}+({\sigma^2_A})^{(0)}}\right)^2\frac{\EE\left\{\left[\sum_i (\eta_i-\bar{\eta})(Y_i-\bar{Y})\right]^2\,|\,\mathcal{G}_{\theta}\right\}}{(n-1)^2}.
	\end{split}
	\]
	Note that  $\{\eta_i\}$ are i.i.d conditional on $\mathcal{G}_{\theta}$, we know
	\[
	\EE\left\{\left[\frac{\sum_i(\eta_i-\bar{\eta})^2}{n-1}\right]^2\,|\,\mathcal{G}_{\theta}\right\}=\left\{\EE\left[\frac{\sum_i(\eta_i-\bar{\eta})^2}{n-1}\,|\,\mathcal{G}_{\theta}\right]\right\}^2+\bigO(1/n).
	\] That is, $\var\left[\frac{\sum_i(\eta_i-\bar{\eta})^2}{n-1}\,|\,\mathcal{G}_{\theta}\right]=\bigO(1/n)$.
	Finally, the term
	\[\label{tmp_eq3}
	\begin{split}
	&\frac{\EE\left\{\left[\sum_i (\eta_i-\bar{\eta})(Y_i-\bar{Y})\right]^2\,|\,\mathcal{G}_{\theta}\right\}}{(n-1)^2}\\
	&=\frac{\EE\left\{\left[\sum_i (\eta_i-\bar{\eta})^2(Y_i-\bar{Y})^2\right]\,|\,\mathcal{G}_{\theta}\right\}+\EE[\bar{\eta}^2\,|\,\mathcal{G}_{\theta}]\sum_{i\neq j}(Y_i-\bar{Y})(Y_j-\bar{Y})}{(n-1)^2}\\
	&=\frac{\sum_i(Y_i-\bar{Y})^2}{(n-1)^2}\EE\left[(\eta_1-\bar{\eta})^2\,|\,\mathcal{G}_{\theta}\right]+\bigO(1/n)\\
	&=\frac{\Delta}{(n-1)^2}\frac{(n-1)\frac{({\sigma^2_A})^{(0)}{\sigma_V^2}}{{\sigma_V^2}+({\sigma^2_A})^{(0)}}}{n}+\bigO(1/n)=\bigO(1/n).
	\end{split}
	\]
	Therefore, we have
	$\var[S\,|\,\mathcal{G}_{\theta}]=\bigO(1/n)$.
\end{proof}

Next, using the following results
\[
\begin{split}
\EE[S\,|\,\mathcal{G}_{\theta}]&=\frac{({\sigma^2_A})^{(0)}{\sigma_V^2}}{{\sigma_V^2}+({\sigma^2_A})^{(0)}}+\left(\frac{({\sigma^2_A})^{(0)}}{{\sigma_V^2}+({\sigma^2_A})^{(0)}}\right)^2\frac{\Delta}{n-1}\\
&\le {\sigma_V^2}+\left(\frac{({\sigma^2_A})^{(0)}}{{\sigma_V^2}+({\sigma^2_A})^{(0)}}\right)^2\frac{\Delta}{n-1}=\bigO(1),\\
\EE[S^2\,|\,\mathcal{G}_{\theta}]&=\left(\EE[S\,|\,\mathcal{G}_{\theta}]\right)^2+\bigO(1/n)=\bigO(1),
\end{split}
\]
we can first write $f''(x^{(1)})$ by
\[
\begin{split}
f''(x^{(1)})=
&n\EE\left[(\bar{\theta}^{(1)}-\bar{Y})^2\,|\,\mathcal{G}_{\theta}\right]+n\left(\frac{\Delta}{n-1}-{\sigma_V^2}\right)^2 \\
&\quad+n \EE[S^2\,|\,\mathcal{G}_{\theta}]-2n\left(\frac{\Delta}{n-1}-{\sigma_V^2}\right)\EE[S\,|\,\mathcal{G}_{\theta}]+\bigO(1).
\end{split}
\]
Then, using
\[
\begin{split}
n\EE\left[ (\bar{\theta}^{(1)}-\bar{Y})^2\,|\,\mathcal{G}_{\theta}\right]
&=\frac{({\sigma^2_A})^{(0)}{\sigma_V^2}}{{\sigma_V^2}+({\sigma^2_A})^{(0)}}+n\left(\frac{{\sigma_V^2}}{{\sigma_V^2}+ ({\sigma^2_A})^{(0)}}\right)^2\left(\mu^{(1)}-\bar{Y}\right)^2\\
&\le {\sigma_V^2} +\frac{n({\sigma_V^2})^2\left(\mu^{(1)}-\bar{Y}\right)^2}{({\sigma_V^2}+ ({\sigma^2_A})^{(0)})^2}
\end{split}
\]
we further bound the terms
\[
\begin{split}
&n\EE\left[(\bar{\theta}^{(1)}-\bar{Y})^2\,|\,\mathcal{G}_{\theta}\right]+n\left(\frac{\Delta}{n-1}-{\sigma_V^2}\right)^2 \\
&\quad+n \EE[S^2\,|\,\mathcal{G}_{\theta}]-2n\left(\frac{\Delta}{n-1}-{\sigma_V^2}\right)\EE[S\,|\,\mathcal{G}_{\theta}]\\
&\le \frac{n({\sigma_V^2})^2\left(\mu^{(1)}-\bar{Y}\right)^2}{({\sigma_V^2}+ ({\sigma^2_A})^{(0)})^2}+n\left[\left(\frac{\Delta}{n-1}-{\sigma_V^2}\right)-\EE[S\,|\,\mathcal{G}_{\theta}]\right]^2\\
&=\frac{n({\sigma_V^2})^2\left(\mu^{(1)}-\bar{Y}\right)^2}{({\sigma_V^2}+ ({\sigma^2_A})^{(0)})^2}+n\left[\frac{({\sigma^2_A})^{(0)}{\sigma_V^2}}{{\sigma_V^2}+({\sigma^2_A})^{(0)}}+\left(\frac{({\sigma^2_A})^{(0)}}{{\sigma_V^2}+({\sigma^2_A})^{(0)}}\right)^2\frac{\Delta}{n-1}-\left(\frac{\Delta}{n-1}-{\sigma_V^2}\right)\right]^2\\
&=\frac{n({\sigma_V^2})^2\left(\mu^{(1)}-\bar{Y}\right)^2}{({\sigma_V^2}+ ({\sigma^2_A})^{(0)})^2}+n\left[\frac{\Delta}{n-1}\left[\left(\frac{({\sigma^2_A})^{(0)}}{{\sigma_V^2}+({\sigma^2_A})^{(0)}}\right)^2-1\right]+\left(\frac{({\sigma^2_A})^{(0)}{\sigma_V^2}}{{\sigma_V^2}+({\sigma^2_A})^{(0)}}+{\sigma_V^2}\right)\right]^2\\
&=\frac{n({\sigma_V^2})^2\left(\mu^{(1)}-\bar{Y}\right)^2}{({\sigma_V^2}+ ({\sigma^2_A})^{(0)})^2}+n\left(\frac{({\sigma^2_A})^{(0)}}{{\sigma_V^2}+({\sigma^2_A})^{(0)}}+1\right)^2\left[\frac{\Delta}{n-1}\left(\frac{-{\sigma_V^2}}{{\sigma_V^2}+({\sigma^2_A})^{(0)}}\right)+{\sigma_V^2}\right]^2\\
&=\frac{n({\sigma_V^2})^2\left(\mu^{(1)}-\bar{Y}\right)^2}{({\sigma_V^2}+ ({\sigma^2_A})^{(0)})^2}+\frac{n({\sigma_V^2})^2({\sigma_V^2}+2({\sigma^2_A})^{(0)})^2}{({\sigma_V^2}+({\sigma^2_A})^{(0)})^4}\left[\frac{\Delta}{n-1}-(({\sigma^2_A})^{(0)}+{\sigma_V^2})\right]^2.
\end{split}
\]
Finally, combing all the results yields
\[
f''(x^{(1)})
&=\frac{n({\sigma_V^2})^2\left(\mu^{(1)}-\bar{Y}\right)^2}{({\sigma_V^2}+ ({\sigma^2_A})^{(0)})^2}+\frac{n({\sigma_V^2})^2({\sigma_V^2}+2({\sigma^2_A})^{(0)})^2}{({\sigma_V^2}+({\sigma^2_A})^{(0)})^4}\left[\frac{\Delta}{n-1}-(({\sigma^2_A})^{(0)}+{\sigma_V^2})\right]^2+\bigO(1).
\]
\subsection{Compute $\EE[f(x^{(1)})\,|\, x^{(0)}]=\EE[f''(x^{(1)})\,|\, \mathcal{G}_{\mu}]$}
Recall that the expectation $\EE[\cdot\,|\,\mathcal{G}_{\mu}]$ is over
\[
\mu^{(1)}&\sim \mathcal{N}\left(\bar{\theta}^{(0)},\frac{({\sigma^2_A})^{(0)}}{n}\right).
\]
In the obtained expression of $f''(x^{(1)})$ from previous step, the
only term involves $\mu^{(1)}$ is $\frac{n({\sigma_V^2})^2\left(\mu^{(1)}-\bar{Y}\right)^2}{({\sigma_V^2}+ ({\sigma^2_A})^{(0)})^2}$.
Since
\[
\EE\left[(\mu^{(1)}-\bar{Y})^2\,|\,\mathcal{G}_{\mu}\right]=(\bar{\theta}^{(0)}-\bar{Y})^2+({\sigma^2_A})^{(0)}/n,
\]
we have
\[
\begin{split}
\EE[f(x^{(1)})\,|\, x^{(0)}]&=\EE[f''(x^{(1)})\,|\, \mathcal{G}_{\mu}]\\
&\le\frac{n({\sigma_V^2})^2}{({\sigma_V^2}+ ({\sigma^2_A})^{(0)})^2}\left((\bar{\theta}^{(0)}-\bar{Y})^2+\frac{({\sigma^2_A})^{(0)}}{n}\right)\\
&\quad+\frac{n({\sigma_V^2})^2({\sigma_V^2}+2({\sigma^2_A})^{(0)})^2}{({\sigma_V^2}+({\sigma^2_A})^{(0)})^4}\left[\frac{\Delta}{n-1}-(({\sigma^2_A})^{(0)}+{\sigma_V^2})\right]^2+\bigO(1)\\
&=\frac{n({\sigma_V^2})^2(\bar{\theta}^{(0)}-\bar{Y})^2}{({\sigma_V^2}+ ({\sigma^2_A})^{(0)})^2}\\
&\quad+\frac{n({\sigma_V^2})^2({\sigma_V^2}+2({\sigma^2_A})^{(0)})^2}{({\sigma_V^2}+({\sigma^2_A})^{(0)})^4}\left[\frac{\Delta}{n-1}-(({\sigma^2_A})^{(0)}+{\sigma_V^2})\right]^2+\bigO(1).
\end{split}
\]
Finally, we complete the proof by
\[
\begin{split}
&\frac{n({\sigma_V^2})^2(\bar{\theta}^{(0)}-\bar{Y})^2}{({\sigma_V^2}+ ({\sigma^2_A})^{(0)})^2}+\frac{n({\sigma_V^2})^2({\sigma_V^2}+2({\sigma^2_A})^{(0)})^2}{({\sigma_V^2}+({\sigma^2_A})^{(0)})^4}\left[\frac{\Delta}{n-1}-(({\sigma^2_A})^{(0)}+{\sigma_V^2})\right]^2+\bigO(1)\\
&=\frac{n({\sigma_V^2})^2({\sigma_V^2}+2({\sigma^2_A})^{(0)})^2}{({\sigma_V^2}+({\sigma^2_A})^{(0)})^4}\left\{\frac{({\sigma_V^2}+ ({\sigma^2_A})^{(0)})^2}{({\sigma_V^2}+2({\sigma^2_A})^{(0)})^2}(\bar{\theta}^{(0)}-\bar{Y})^2+\left[\frac{\Delta}{n-1}-(({\sigma^2_A})^{(0)}+{\sigma_V^2})\right]^2\right\}+\bigO(1)\\
&\le \frac{({\sigma_V^2})^2({\sigma_V^2}+2({\sigma^2_A})^{(0)})^2}{({\sigma_V^2}+({\sigma^2_A})^{(0)})^4}\left\{n(\bar{\theta}^{(0)}-\bar{Y})^2+n\left[\frac{\Delta}{n-1}-(({\sigma^2_A})^{(0)}+{\sigma_V^2})\right]^2\right\}+\bigO(1)\\
&=\left[\left(\frac{({\sigma_V^2})^2+2{\sigma_V^2}({\sigma^2_A})^{(0)}}{({\sigma_V^2})^2+2{\sigma_V^2}({\sigma^2_A})^{(0)}+(({\sigma^2_A})^{(0)})^2}\right)^2\right] f(x^{(0)})+\bigO(1).
\end{split}
\]

\end{proof}

	\section{Proof of Lemma \ref{lemma_epsilon}}\label{proof_lemma_epsilon}
	
	\begin{lemma}\label{lemma_epsilon}
Under the assumptions of \cref{thm_Gibbs}, recall the ``large set'' defined in the proof of \cref{thm_Gibbs}.
	If $T=\Theta(1)$, by choosing the size of the ``small set'' $R=\{x\in\mathcal{X}: f_n(x)\le d\}$ to satisfy $d=\bigO(1)$ and $d> \frac{b}{1-\lambda_T}$, {there exists a probability measure $Q(\cdot)$ such that} the Markov chain satisfies a minorization condition in \cref{eq_minorization} with the minorization {volumne} $\epsilon=\Theta(1)$. 
\end{lemma}

\begin{proof}
{Throughout the proof, we write $f_n(x)$ as $f(x)$ for simplicity.} Recall that the small set is defined by $R=\{x\in\mathcal{X}: f(x)\le d\}$ where $d>2b/(1-\lambda_T)$ and $x=({\sigma^2_A}, \mu,\theta_1,\dots,\theta_n)$. When $b=\bigO(1)$ and $\lambda_T=\Theta(1)$, we can choose $d=\bigO(1)$. Our goal is to show the minorization volume $\epsilon$ satisfying
\[
P(x,\cdot)\ge \epsilon Q(\cdot), \quad \forall x\in R,
\]
is asymptotically bounded away from $0$. Denoting ${\hat{\sigma}^2_A}:=\frac{\Delta}{n-1}-{\sigma_V^2}$, we have
\[
\begin{split}
R&=\left\{x\in\mathcal{X}: n(\bar{\theta}-\bar{Y})^2+n\left[\left(\frac{\Delta}{n-1}-{\sigma_V^2}\right)-{\sigma^2_A}\right]^2\le d\right\}\\
&\subseteq \left\{x\in\mathcal{X}:|\bar{\theta}-\bar{Y}|\le \sqrt{\frac{d}{n}} \right\}\bigcap\left\{x\in\mathcal{X}:|{\sigma^2_A}-{{\hat{\sigma}^2_A}}|\le \sqrt{\frac{d}{n}}\right\}
\end{split}
\]
Denoting 
\[
R':=\left\{x\in\mathcal{X}:|\bar{\theta}-\bar{Y}|\le \sqrt{\frac{d}{n}}, |{\sigma^2_A}-{{\hat{\sigma}^2_A}}|\le \sqrt{\frac{d}{n}} \right\}
\]
since $R\subseteq R'$, it suffices to show the minorization volume $\epsilon$ satisfying
\[
P(x^{(0)},\cdot)\ge \epsilon Q(\cdot), \quad \forall x^{(0)}\in R',
\]
is asymptotically bounded away from $0$. One common technique to obtain $\epsilon$ is by integrating the infimum of densities of $P(x^{(0)},\cdot)$ where in our case the infimum is over all $\bar{\theta}^{(0)}$ and $({\sigma^2_A})^{(0)}$ such that $|\bar{\theta}^{(0)}-\bar{Y}|\le \sqrt{\frac{d}{n}}$ and $|({\sigma^2_A})^{(0)}-{{\hat{\sigma}^2_A}}|\le \sqrt{\frac{d}{n}}$. 

Note that the intuition behind the proof is: since $R'$ is determined by $|\bar{\theta}^{(0)}-\bar{Y}|\le \sqrt{\frac{d}{n}}$ and $|({\sigma^2_A})^{(0)}-{{\hat{\sigma}^2_A}}|\le \sqrt{\frac{d}{n}}$. The size of uncertainties of the initial $\bar{\theta}^{(0)}$ and $({\sigma^2_A})^{(0)}$ is of order $\bigO(1/\sqrt{n})$. Therefore, for any fixed initial state $x^{(0)}\in R'$, if the transition kernel $P(x^{(0)},\cdot)$ concentrates at a rate of $\Omega(1/\sqrt{n})$ then $\epsilon$ is bounded away from $0$.

For the density function of the Markov transition kernel $P(x^{(0)},\cdot)$, recall the order of Gibbs sampler
\[
\begin{split}
\mu^{(1)}&\sim \mathcal{N}\left(\bar{\theta}^{(0)},\frac{({\sigma^2_A})^{(0)}}{n}\right),\\
\theta_i^{(1)}&\sim \mathcal{N}\left(\frac{\mu^{(1)}{\sigma_V^2}+Y_i ({\sigma^2_A})^{(0)}}{{\sigma_V^2}+ ({\sigma^2_A})^{(0)}},\frac{({\sigma^2_A})^{(0)}{\sigma_V^2}}{{\sigma_V^2}+({\sigma^2_A})^{(0)}}\right),\quad i=1,\dots,n\\
({\sigma^2_A})^{(1)}&\sim \IG\left(a+\frac{n-1}{2},b+\frac{1}{2}\sum_{i=1}^n (\theta_i^{(1)}-\bar{\theta}^{(1)})^2\right).
\end{split}
\]
Then $\epsilon$ can be computed using the three steps of integration according to the reverse order of the Gibbs sampler:
\begin{enumerate}
	\item For given $\mu^{(1)}$ and $\{\theta_i^{(1)}\}$, integrating the infimum of the density of $({\sigma^2_A})^{(1)}$. Note that the infimum is over a subset of $\bar{\theta}^{(0)}$ and $({\sigma^2_A})^{(0)}$. However, 
	\[
	({\sigma^2_A})^{(1)}&\sim \IG\left(a+\frac{n-1}{2},b+\frac{1}{2}\sum_{i=1}^n (\theta_i^{(1)}-\bar{\theta}^{(1)})^2\right)
	\]
	does not depend on $\bar{\theta}^{(0)}$ and $({\sigma^2_A})^{(0)}$. Therefore,
	the integration of the infimum of the density in this step always equals one;
	\item For given $\mu^{(1)}$, integrating the infimum of the densities of $\{\theta_i^{(1)}\}$. We first note that $\{\theta_i^{(1)}\}$ appear in the densities only in the forms of  $\bar{\theta}^{(1)}$ and $S=\frac{\sum_i (\theta_i^{(1)}-\bar{\theta}^{(1)})^2}{n-1}$. Therefore, instead of integrating over $(\theta_1^{(1)},\dots,\theta_n^{(1)})$ we can integrate over $\bar{\theta}^{(1)}$ and $S$. Furthermore, we have shown $\bar{\theta}^{(1)}$ is conditional independent with $S$ given $({\sigma^2_A})^{(0)}$ in the proof of \cref{temp_lemma}, we can integrate them separately. Finally, we note that the infimum is over $\left\{({\sigma^2_A})^{(0)}: |({\sigma^2_A})^{(0)}-{{\hat{\sigma}^2_A}}|\le \sqrt{\frac{d}{n}}\right\}$. Overall, we need to show $\tilde{g}_n(\mu^{(1)})$ is lower bounded away from $0$, which is defined by
	\[
	\begin{split}
	\tilde{g}_n(\mu^{(1)})&:=\int \dee S\dee \bar{\theta} \inf_{x^{(0)}\in R'} \left\{f_S(({\sigma^2_A})^{(0)},n; S)\,
	\Normal\left(\frac{\mu^{(1)}{\sigma_V^2}+\bar{Y} ({\sigma^2_A})^{(0)}}{{\sigma_V^2}+ ({\sigma^2_A})^{(0)}},\frac{({\sigma^2_A})^{(0)}{\sigma_V^2}}{n({\sigma_V^2}+({\sigma^2_A})^{(0)})};\bar{\theta}\right)\right\}\\
	&\ge\left[\int \dee S \inf_{x^{(0)}\in R'} f_S(({\sigma^2_A})^{(0)},n; S)\right]
	\\
	&\qquad \cdot\left[\int\dee \bar{\theta} \inf_{x^{(0)}\in R'}\Normal\left(\frac{\mu^{(1)}{\sigma_V^2}+\bar{Y} ({\sigma^2_A})^{(0)}}{{\sigma_V^2}+ ({\sigma^2_A})^{(0)}},\frac{({\sigma^2_A})^{(0)}{\sigma_V^2}}{n({\sigma_V^2}+({\sigma^2_A})^{(0)})};\bar{\theta}\right)\right],
	\end{split}
	\]
	where $f_S(({\sigma^2_A})^{(0)},n;S)$ denotes the density function of $S=\frac{\sum_i (\theta_i-\bar{\theta})^2}{n-1}$ for given $({\sigma^2_A})^{(0)}$, with
	\[
	\theta_i\sim \mathcal{N}\left(\frac{\mu^{(1)}{\sigma_V^2}+Y_i ({\sigma^2_A})^{(0)}}{{\sigma_V^2}+ ({\sigma^2_A})^{(0)}},\frac{({\sigma^2_A})^{(0)}{\sigma_V^2}}{{\sigma_V^2}+({\sigma^2_A})^{(0)}}\right),\quad i=1,\dots,n,
	\]
	and $\Normal\left(\frac{\mu^{(1)}{\sigma_V^2}+\bar{Y} ({\sigma^2_A})^{(0)}}{{\sigma_V^2}+ ({\sigma^2_A})^{(0)}},\frac{({\sigma^2_A})^{(0)}{\sigma_V^2}}{n({\sigma_V^2}+({\sigma^2_A})^{(0)})};\bar{\theta}\right)$ denotes the density function of 
	\[
	\bar{\theta}\sim \Normal\left(\frac{\mu^{(1)}{\sigma_V^2}+\bar{Y} ({\sigma^2_A})^{(0)}}{{\sigma_V^2}+ ({\sigma^2_A})^{(0)}},\frac{({\sigma^2_A})^{(0)}{\sigma_V^2}}{n({\sigma_V^2}+({\sigma^2_A})^{(0)})}\right).
	\]
	\item Finally, we integrate the infimum of the densities of $\mu^{(1)}$ to get $\epsilon$. That is,
	\[
	\epsilon=\int \dee \mu  \left\{\tilde{g}_n(\mu)\inf_{x^{(0)}\in R'} \mathcal{N}\left(\bar{\theta}^{(0)}, \frac{({\sigma^2_A})^{(0)}}{n}; \mu\right)\right\}.
	\]
\end{enumerate}
In the following, we show $\epsilon$ is lower bounded away from $0$ in three steps. 
	
	First, it is easy to see that the density of $S$ does not depend on $\mu^{(1)}$. We show
	\[\label{tmp_step1}
	\int \dee S \inf_{x^{(0)}\in R'} f_S(({\sigma^2_A})^{(0)},n; S)=\Theta(1).
	\]
	
	Second, we show 
	\[\label{tmp_step2}
	\int\dee \bar{\theta} \inf_{x^{(0)}\in R'}\Normal\left(\frac{\mu^{(1)}{\sigma_V^2}+\bar{Y} ({\sigma^2_A})^{(0)}}{{\sigma_V^2}+ ({\sigma^2_A})^{(0)}},\frac{({\sigma^2_A})^{(0)}{\sigma_V^2}}{n({\sigma_V^2}+({\sigma^2_A})^{(0)})};\bar{\theta}\right)\ge 
	1-\textrm{erf}\left(\frac{C|\mu|+C'}{\sqrt{2}}\right)
	\]
	where $\textrm{erf}(z):=\frac{2}{\sqrt{\pi}}\int_{0}^ze^{-t^2}\dee t$ and $C$ and $C'$ are some constants.
	
	Finally, we complete the proof by showing
	\[\label{tmp_step3}
	\int \dee \mu  \left\{\left(1-\textrm{erf}(\frac{C|\mu|+C'}{\sqrt{2}})\right)\inf_{x^{(0)}\in R'} \mathcal{N}\left(\bar{\theta}^{(0)}, \frac{({\sigma^2_A})^{(0)}}{n}; \mu\right)\right\}=\Theta(1).
	\]

\subsection{Proof of \cref{tmp_step1}}
We omit the superscripts for simplicity. That is, we show
\[
\int \dee S \inf_{\left\{{\sigma^2_A}: |{\sigma^2_A}-{{\hat{\sigma}^2_A}}|\le \sqrt{\frac{d}{n}}\right\}} f_S({\sigma^2_A},n; S)=\Theta(1).
\]
Following the proof of \cref{temp_lemma} from \cref{tmp_eq2} to \cref{tmp_eq3}, defining
\[
\eta_i:=\theta_i-\frac{Y_i {\sigma^2_A}}{{\sigma_V^2}+{\sigma^2_A}}\sim \mathcal{N}\left(\frac{\mu {\sigma_V^2}}{{\sigma_V^2}+{\sigma^2_A}},\frac{{\sigma^2_A} {\sigma_V^2}}{{\sigma_V^2}+{\sigma^2_A}}\right),
\]
we know
\[
\EE\left[\left|S-\frac{\sum_i (\eta_i-\bar{\eta})^2}{n-1}-\left(\frac{{\sigma^2_A}}{{\sigma_V^2}+{\sigma^2_A}}\right)^2\frac{\Delta}{n-1}\right|^2\right]=\bigO(1/n).
\]
Therefore, defining 
\[
S':=\frac{\sum_i (\eta_i-\bar{\eta})^2}{n-1}+\left(\frac{{\sigma^2_A}}{{\sigma_V^2}+{\sigma^2_A}}\right)^2\frac{\Delta}{n-1}
\]
and denoting $f_{S'}'({\sigma^2_A},n; S')$ as the density of $S'$, it suffices to show
\[
\int \dee S' \inf_{\left\{{\sigma^2_A}: |{\sigma^2_A}-{{\hat{\sigma}^2_A}}|\le \sqrt{\frac{d}{n}}\right\}} f_{S'}'({\sigma^2_A},n; S')=\Theta(1).
\]
Furthermore, note that under $|{\sigma^2_A}-{{\hat{\sigma}^2_A}}|\le \sqrt{\frac{d}{n}}$, we have $\frac{{\sigma_V^2}+{\sigma^2_A}}{{\sigma^2_A} {\sigma_V^2}}=\frac{{\sigma_V^2}+{{\hat{\sigma}^2_A}}}{{{\hat{\sigma}^2_A}} {\sigma_V^2}}+\bigO(1/\sqrt{n})=\Theta(1)$. Then it suffices to show
\[
\int \dee S'' \inf_{\left\{{\sigma^2_A}: |{\sigma^2_A}-{{\hat{\sigma}^2_A}}|\le \sqrt{\frac{d}{n}}\right\}} f_{S''}''({\sigma^2_A},n; S'')=\Theta(1),
\]
where 
\[
S'':&=\frac{{\sigma_V^2}+{\sigma^2_A}}{{\sigma^2_A} {\sigma_V^2}}S'
=\frac{{\sigma_V^2}+{\sigma^2_A}}{{\sigma^2_A} {\sigma_V^2}}\frac{\sum_i (\eta_i-\bar{\eta})^2}{n-1}+\frac{1}{{\sigma_V^2}}\left(\frac{{\sigma^2_A}}{{\sigma_V^2}+{\sigma^2_A}}\right)\frac{\Delta}{n-1}
\] 
and $f_{S''}''({\sigma^2_A},n; S'')$ is the density function of $S''$.

Next, note that $\frac{{\sigma_V^2}+{\sigma^2_A}}{{\sigma^2_A} {\sigma_V^2}}\sum_i (\eta_i-\bar{\eta})^2\sim \chi^2_{n-1}$, we have 
\[
\frac{\frac{{\sigma_V^2}+{\sigma^2_A}}{{\sigma^2_A} {\sigma_V^2}}\sum_i (\eta_i-\bar{\eta})^2-(n-1)}{\sqrt{2(n-1)}}\xrightarrow{d} \Normal(0,1),
\] 
which does not depend on $n$.
We define $\tilde{f}(z, {\sigma^2_A}; x), \forall z\in\mathbb{R}$ as the density function of a random variable 
\[
\tilde{X}_{z,{\sigma^2_A}}:=z+\frac{\frac{{\sigma_V^2}+{\sigma^2_A}}{{\sigma^2_A} {\sigma_V^2}}\sum_i (\eta_i-\bar{\eta})^2
	-(n-1)}{\sqrt{2(n-1)}},
\]
then we know $\tilde{X}_{z,{\sigma^2_A}}\xrightarrow{d} \Normal(z,1)$.

The rest of the proof is first to lower bound
$\int \dee S'' \inf_{\left\{{\sigma^2_A}: |{\sigma^2_A}-{{\hat{\sigma}^2_A}}|\le \sqrt{\frac{d}{n}}\right\}} f_{S''}''({\sigma^2_A},n; S'')$ 
using the density function $\tilde{f}(z, {\sigma^2_A}; x)$ and then show it is asymptotically lower bounded away from $0$.

Notice that $\frac{1}{{\sigma_V^2}}\left(\frac{{\sigma^2_A}}{{\sigma_V^2}+{\sigma^2_A}}\right)\frac{\Delta}{n-1}$ is not random, and there exists a constant $C_0$ such that
\[
&\left(\max_{\{{\sigma^2_A}: |{\sigma^2_A}-{{\hat{\sigma}^2_A}}|\le \sqrt{d/n}\}}\frac{{\sigma^2_A}}{{\sigma_V^2}+{\sigma^2_A}}-\min_{\{{\sigma^2_A}: |{\sigma^2_A}-{{\hat{\sigma}^2_A}}|\le \sqrt{d/n}\}}\frac{{\sigma^2_A}}{{\sigma_V^2}+{\sigma^2_A}}\right)\frac{\Delta/{\sigma_V^2}}{n-1}\le\frac{C_0}{\sqrt{n-1}}.
\]
Finally we have
\[
\begin{split}
&\int \dee S'' \inf_{\left\{{\sigma^2_A}: |{\sigma^2_A}-{{\hat{\sigma}^2_A}}|\le \sqrt{\frac{d}{n}}\right\}} f_{S''}''({\sigma^2_A},n; S'')\\
&\ge \inf_{\left\{{\sigma^2_A}: |{\sigma^2_A}-{{\hat{\sigma}^2_A}}|\le \sqrt{\frac{d}{n}}\right\}}\int \dee x \min\left\{ \tilde{f}\left(-\frac{C_0}{\sqrt{2}}, {\sigma^2_A}; x\right), \tilde{f}\left(+\frac{C_0}{\sqrt{2}}, {\sigma^2_A}; x\right) \right\}\\
&= 1- \sup_{\{{\sigma^2_A}: |{\sigma^2_A}-{{\hat{\sigma}^2_A}}|\le \sqrt{d/n}\}}\int_{-\sqrt{2}C_0}^{\sqrt{2}C_0} \dee x \tilde{f}(0, {\sigma^2_A}; x)\\
&=1-\sup_{\{{\sigma^2_A}: |{\sigma^2_A}-{{\hat{\sigma}^2_A}}|\le \sqrt{d/n}\}}\Pr(-\sqrt{2}C_0\le \tilde{X}_{0,{\sigma^2_A}}\le \sqrt{2}C_0)\\
&\to 1-\int_{-\sqrt{2}C_0}^{\sqrt{2}C_0} \dee x \Normal(0,1; x)=\Theta(1).
\end{split}
\]

\subsection{Proof of \cref{tmp_step2}}
We again omit the subscripts for simplicity. The goal is to lower bound
\[
	\int\dee \bar{\theta} \inf_{\left\{{\sigma^2_A}: |{\sigma^2_A}-{{\hat{\sigma}^2_A}}|\le \sqrt{\frac{d}{n}}\right\}}\Normal\left(\frac{\mu {\sigma_V^2}+\bar{Y} {\sigma^2_A}}{{\sigma_V^2}+ {\sigma^2_A}},\frac{{\sigma^2_A} {\sigma_V^2}}{n({\sigma_V^2}+{\sigma^2_A})};\bar{\theta}\right)
\]
Note that there exists some constants $C_1$ and $C_2$ such that
\[
\max_{\left\{{\sigma^2_A}: |{\sigma^2_A}-{{\hat{\sigma}^2_A}}|\le \sqrt{\frac{d}{n}}\right\}} \frac{\mu {\sigma_V^2}+\bar{Y}{\sigma^2_A}}{{\sigma_V^2}+{\sigma^2_A}}-\min_{\left\{{\sigma^2_A}: |{\sigma^2_A}-{{\hat{\sigma}^2_A}}|\le \sqrt{\frac{d}{n}}\right\}} \frac{\mu {\sigma_V^2}+\bar{Y}{\sigma^2_A}}{{\sigma_V^2}+{\sigma^2_A}}\le
\frac{C_1|\mu|+C_2}{\sqrt{n}},
\]
and another constant $C_3$ such that
\[
\min_{\left\{{\sigma^2_A}: |{\sigma^2_A}-{{\hat{\sigma}^2_A}}|\le \sqrt{\frac{d}{n}}\right\}} \frac{{\sigma^2_A} {\sigma_V^2}}{n({\sigma_V^2}+{\sigma^2_A})} \ge \frac{C_3}{n}.
\]
Therefore, we have
\[
\begin{split}
	&\int\dee \bar{\theta} \inf_{\left\{{\sigma^2_A}: |{\sigma^2_A}-{{\hat{\sigma}^2_A}}|\le \sqrt{\frac{d}{n}}\right\}}\Normal\left(\frac{\mu {\sigma_V^2}+\bar{Y} {\sigma^2_A}}{{\sigma_V^2}+ {\sigma^2_A}},\frac{{\sigma^2_A} {\sigma_V^2}}{n({\sigma_V^2}+{\sigma^2_A})};\bar{\theta}\right)\\
	&\ge 2 \int_{(C_1|\mu|+C_2)/\sqrt{n}}^{\infty} \dee x\, \Normal(0,C_3/n; x)\\ &=2\int_{C_4|\mu|+C_5}^{\infty} \dee x\, \Normal(0,1; x)\\
	&=1-\textrm{erf}\left(\frac{C_4|\mu|+C_5}{\sqrt{2}}\right),
\end{split}
\]
where $C_4:=\frac{C_1}{\sqrt{C_3}}$ and $C_5:=\frac{C_2}{\sqrt{C_3}}$.

\subsection{Proof of \cref{tmp_step3}}
We omit the subscripts for simplicity. We show the following is asymptotically bounded away from $0$:
\[
	\int \dee \mu  \left\{\left(1-\textrm{erf}\left(\frac{C_4|\mu|+C_5}{\sqrt{2}}\right)\right)\inf_{x\in R'} \mathcal{N}\left(\bar{\theta}, \frac{{\sigma^2_A}}{n}; \mu\right)\right\}
\]
Note that there exists $({\sigma^2_A})_n'\in [{{\hat{\sigma}^2_A}}-\sqrt{d/n},{{\hat{\sigma}^2_A}}+\sqrt{d/n}]$ such that
\[
\begin{split}
&\inf_{\left\{(\bar{\theta},{\sigma^2_A}): |\bar{\theta}-\bar{Y}|\le \sqrt{\frac{d}{n}}, |{\sigma^2_A}-{{\hat{\sigma}^2_A}}|\le \sqrt{\frac{d}{n}}\right\}} \mathcal{N}\left(\bar{\theta}, \frac{{\sigma^2_A}}{n}; \mu\right)\\
&\quad=\min\left\{\Normal\left(\bar{Y}-\sqrt{\frac{d}{n}}, \frac{({\sigma^2_A})_n'}{n};\mu\right),\Normal\left(\bar{Y}+\sqrt{\frac{d}{n}}, \frac{({\sigma^2_A})_n'}{n};\mu\right)\right\}
\end{split}
\]

Therefore, we have
\[
\begin{split}
&\int_{-\infty}^{\infty} \dee \mu  \left\{\left(1-\textrm{erf}\left(\frac{C_4|\mu|+C_5}{\sqrt{2}}\right)\right)\inf_{\left\{(\bar{\theta},{\sigma^2_A}): |\bar{\theta}-\bar{Y}|\le \sqrt{\frac{d}{n}}, |{\sigma^2_A}-{{\hat{\sigma}^2_A}}|\le \sqrt{\frac{d}{n}}\right\}} \mathcal{N}\left(\bar{\theta}, \frac{{\sigma^2_A}}{n}; \mu\right)\right\}\\
&\ge \int_{0}^{2\bar{Y}} \dee \mu  \left\{\left(1-\textrm{erf}\left(\frac{C_4|\mu|+C_5}{\sqrt{2}}\right)\right)\inf_{\left\{(\bar{\theta},{\sigma^2_A}): |\bar{\theta}-\bar{Y}|\le \sqrt{\frac{d}{n}}, |{\sigma^2_A}-{{\hat{\sigma}^2_A}}|\le \sqrt{\frac{d}{n}}\right\}} \mathcal{N}\left(\bar{\theta}, \frac{{\sigma^2_A}}{n}; \mu\right)\right\}\\
&\ge \left(1-\textrm{erf}\left(\frac{C_4|2\bar{Y}|+C_5}{\sqrt{2}}\right)\right)\int_{0}^{2\bar{Y}} \dee \mu \inf_{\left\{(\bar{\theta},{\sigma^2_A}): |\bar{\theta}-\bar{Y}|\le \sqrt{\frac{d}{n}}, |{\sigma^2_A}-{{\hat{\sigma}^2_A}}|\le \sqrt{\frac{d}{n}}\right\}} \mathcal{N}\left(\bar{\theta}, \frac{{\sigma^2_A}}{n}; \mu\right)\\
&=\left(1-\textrm{erf}\left(\frac{C_4|2\bar{Y}|+C_5}{\sqrt{2}}\right)\right)\\
&\quad\cdot\left[\int_{0}^{\bar{Y}}\dee \mu\,\Normal\left(\bar{Y}+\sqrt{\frac{d}{n}}, \frac{({\sigma^2_A})_n'}{n};\mu\right)+\int_{\bar{Y}}^{2\bar{Y}}\dee \mu\,\Normal\left(\bar{Y}-\sqrt{\frac{d}{n}}, \frac{({\sigma^2_A})_n'}{n};\mu\right)\right]\\
&=\left(1-\textrm{erf}\left(\frac{C_4|2\bar{Y}|+C_5}{\sqrt{2}}\right)\right)\\
&\quad\cdot\left[\int_{-\bar{Y}}^{0}\dee \mu\,\Normal\left(\sqrt{\frac{d}{n}}, \frac{({\sigma^2_A})_n'}{n};\mu\right)+\int_{0}^{\bar{Y}}\dee \mu\,\Normal\left(-\sqrt{\frac{d}{n}}, \frac{({\sigma^2_A})_n'}{n};\mu\right)\right]
\end{split}
\]
Finally, we show
\[
\int_{-\bar{Y}}^{0}\dee \mu\,\Normal\left(\sqrt{\frac{d}{n}}, \frac{({\sigma^2_A})_n'}{n};\mu\right)+\int_{0}^{\bar{Y}}\dee \mu\,\Normal\left(-\sqrt{\frac{d}{n}}, \frac{({\sigma^2_A})_n'}{n};\mu\right)
\]
is asymptotically bounded away from $0$. Note that when $n\to \infty$, we have $({\sigma^2_A})_n'\to {{\hat{\sigma}^2_A}}$. So the density functions $\Normal\left(\pm\sqrt{\frac{d}{n}}, \frac{({\sigma^2_A})_n'}{n};\mu\right)$ concentrate on $0$. Therefore
\[
\begin{split}
&\int_{-\bar{Y}}^{0}\dee \mu\,\Normal\left(\sqrt{\frac{d}{n}}, \frac{({\sigma^2_A})_n'}{n};\mu\right)+\int_{0}^{\bar{Y}}\dee \mu\,\Normal\left(-\sqrt{\frac{d}{n}}, \frac{({\sigma^2_A})_n'}{n};\mu\right)\\
&\to \int_{-\infty}^{0}\dee \mu\,\Normal\left(\sqrt{\frac{d}{n}}, \frac{{{\hat{\sigma}^2_A}}}{n};\mu\right)+\int_{0}^{\infty}\dee \mu\,\Normal\left(-\sqrt{\frac{d}{n}}, \frac{{{\hat{\sigma}^2_A}}}{n};\mu\right)\\
&=1-\int_{-\sqrt{d/n}}^{\sqrt{d/n}}\dee x \, \Normal\left(0,\frac{{{\hat{\sigma}^2_A}}}{n}; x\right)\\
&=1-\int_{-\sqrt{d}}^{\sqrt{d}}\dee x \, \Normal(0,{{\hat{\sigma}^2_A}}; x)=\Theta(1).
\end{split}
\]
\end{proof}

	\section{Proof of Lemma \ref{lemma_R0}}\label{proof_lemma_R0}
	\begin{lemma}\label{lemma_R0}
	Under the assumptions of \cref{thm_Gibbs}, recall the definition of drift function and ``large set'' in the proof of \cref{thm_Gibbs}. With the initial state $x^{(0)}$ given by \cref{initial_state}, there exists a positive integer $N$, which does not depend on $k$, such that for all $n\ge N$, we have
	\[\label{eq_tail_bound2}
	\begin{split}
	&k\,\pi(R_T^c)+ \sum_{i=1}^k P^i (x^{(0)},R_T^c)\\
	&\le   \frac{k}{\sqrt{n}} \frac{\sqrt{b}(2{\sigma_V^2}/\delta+1)}{\left|\left(\frac{\Delta}{n-1}-{\sigma_V^2}\right)-T\right|} + \frac{k(1+k)}{2n}\frac{b}{\left[\left(\frac{\Delta}{n-1}-{\sigma_V^2}\right)-T\right]^2}.
	\end{split}
	\]
\end{lemma}

\begin{proof}
{In this proof, we write $f_n(x)$ as $f(x)$ for simplicity.} We first consider a Markov chain starting from initial state $x^{(0)}$ defined by \cref{initial_state}. By \cref{eq_assumption}, we have $({\sigma^2_A})^{(0)}=\frac{\sum_{i=1}^n (Y_i-\bar{Y})^2}{n-1}-{\sigma_V^2}$ for large enough $n$, which implies $f(x^{(0)})=0$. Therefore, for large enough $n$, we have $\EE(f(x^{(1)}))\le b$ from \cref{key_thm}. Furthermore, we can continue to get upper bounds $\EE(f(x^{(i)}))\le i b$
for all $i=1,\dots,k$. This implies
\[
\EE\left[\left(\left(\frac{\Delta}{n-1}-{\sigma_V^2}\right)-({\sigma^2_A})^{(i)}\right)^2\right]\le i \frac{b}{n},\quad i=1,\dots,k.
\]
By the Markov's inequality, we have
\[
\begin{split}
&\Pr\left(\left|({\sigma^2_A})^{(i)}-\left(\frac{\Delta}{n-1}-{\sigma_V^2}\right)\right|
\ge \left|T-\left(\frac{\Delta}{n-1}-{\sigma_V^2}\right)\right|\right)
\le\frac{i}{n} \frac{b}{\left[T-\left(\frac{\Delta}{n-1}-{\sigma_V^2}\right)\right]^2},
\end{split}
\]
for $i=1,\dots,k$.
Therefore, we have
\[
\sum_{i=1}^k P^i (x^{(0)},R_T^c)\le \frac{b}{\left[T-\left(\frac{\Delta}{n-1}-{\sigma_V^2}\right)\right]^2}\sum_{i=1}^k \frac{i}{n}=\frac{k(1+k)}{2n}\frac{b}{\left[T-\left(\frac{\Delta}{n-1}-{\sigma_V^2}\right)\right]^2}.
\]
Next, we consider a Markov chain starting from $\pi$. According to \cref{key_thm}, we have
\[
\begin{split}
&\EE_{\pi}\left[\left(1-\left(\frac{({\sigma_V^2})^2+2{\sigma_V^2}{\sigma^2_A}}{({\sigma_V^2})^2+2{\sigma_V^2}{\sigma^2_A}+({\sigma^2_A})^2}\right)^2\right)f(x)\right]\\
&=\EE_{\pi}\left[\left(1+\frac{({\sigma_V^2})^2+2{\sigma_V^2}{\sigma^2_A}}{({\sigma_V^2})^2+2{\sigma_V^2}{\sigma^2_A}+({\sigma^2_A})^2}\right)\left(1-\frac{({\sigma_V^2})^2+2{\sigma_V^2}{\sigma^2_A}}{({\sigma_V^2})^2+2{\sigma_V^2}{\sigma^2_A}+({\sigma^2_A})^2}\right)f(x)\right]\\
&=\EE_{\pi}\left[\left(1+\frac{({\sigma_V^2})^2+2{\sigma_V^2}{\sigma^2_A}}{({\sigma_V^2})^2+2{\sigma_V^2}{\sigma^2_A}+({\sigma^2_A})^2}\right)\left(\frac{{\sigma^2_A}}{{\sigma_V^2}+{\sigma^2_A}}\right)^2f(x)\right]
\le b,
\end{split}
\]
where $\EE_{\pi}[\cdot]$ denotes the expectation is over $x\sim \pi(\cdot)$.
Note that by {H\"older's inequality (in the reverse way)}
\[
\begin{split}
&\EE_{\pi}\left[\left(1+\frac{({\sigma_V^2})^2+2{\sigma_V^2}{\sigma^2_A}}{({\sigma_V^2})^2+2{\sigma_V^2}{\sigma^2_A}+({\sigma^2_A})^2}\right)\left(\frac{{\sigma^2_A}}{{\sigma_V^2}+{\sigma^2_A}}\right)^2f(x)\right]\\
&\ge
\EE_{\pi}\left[\left(\frac{{\sigma^2_A}}{{\sigma_V^2}+{\sigma^2_A}}\right)^2f(x)\right]\\
&\ge [\EE_{\pi}(f(x)^{\frac{1}{2}})]^2\left\{\EE_{\pi}\left[\left(\frac{{\sigma^2_A}}{{\sigma_V^2}+{\sigma^2_A}}\right)^{-2}\right]\right\}^{-1}\\
&=[\EE_{\pi}(f(x)^{\frac{1}{2}})]^2 / \EE_{\pi}[(1+{\sigma_V^2}/{\sigma^2_A})^2].
\end{split}
\]
Therefore, we have
\[
\EE_{\pi}(f(x)^{\frac{1}{2}})\le \sqrt{b}\sqrt{1+2{\sigma_V^2}\EE_{\pi}(1/{\sigma^2_A})+({\sigma_V^2})^2\EE_{\pi}(1/({\sigma^2_A})^2)}.
\]
Next, according to \cref{lemma_bounded_A}, we know that $\EE_{\pi}(1/{\sigma^2_A})\le 2/\delta$ and $\EE_{\pi}(1/({\sigma^2_A})^2)\le 2/\delta^2$ for large enough $n$.

More specifically, by \cref{lemma_bounded_A}, we have
$\sqrt{1+2{\sigma_V^2}\EE_{\pi}(1/{\sigma^2_A})+({\sigma_V^2})^2\EE_{\pi}(1/({\sigma^2_A})^2)}\le 1+2{\sigma_V^2}/\delta$ for large enough $n$. Therefore, we get 
\[
\EE_{\pi}\left(\left|\left(\frac{\Delta}{n-1}-{\sigma_V^2}\right)-{\sigma^2_A}\right|\right)\le \sqrt{\frac{b}{n}}(2{\sigma_V^2}/\delta+1).
\]
Thus, by the Markov's inequality
\[
\begin{split}
\pi(R_T^c)&=\Pr_{\pi}\left(\left|\left(\frac{\Delta}{n-1}-{\sigma_V^2}\right)-{\sigma^2_A}\right|\ge \left|\left(\frac{\Delta}{n-1}-{\sigma_V^2}\right)-T\right|\right)\\
&\le \frac{\sqrt{\frac{b}{n}}(2{\sigma_V^2}/\delta+1)}{\left|\left(\frac{\Delta}{n-1}-{\sigma_V^2}\right)-T\right|}.
\end{split}
\]
Finally, we have
\[
\begin{split}
&k\,\pi(R_T^c)+ \sum_{i=1}^k P^i (x^{(0)},R_T^c)\\
&\le   \frac{k}{\sqrt{n}} \frac{\sqrt{b}(2{\sigma_V^2}/\delta+1)}{\left|\left(\frac{\Delta}{n-1}-{\sigma_V^2}\right)-T\right|} + \frac{k(1+k)}{2n}\frac{b}{\left[T-\left(\frac{\Delta}{n-1}-{\sigma_V^2}\right)\right]^2}.
\end{split}
\]
\end{proof}

\begin{lemma}\label{lemma_bounded_A}
	There exists a positive integer $N$, which only depends on $a$, $b$, ${\sigma_V^2}$, and $\delta$, such that for all $n\ge N$, we have
	\[
	\EE_{\pi}(1/{\sigma^2_A})\le 2/\delta,\quad \EE_{\pi}(1/({\sigma^2_A})^2)\le 2/\delta^2.
	\]
\end{lemma}
\begin{proof}
	The posterior distribution can be written as
	\[
	{\pi}(x\,|\,Y_1,\dots,Y_n)=\frac{f_a(x,Y_1,\dots,Y_n)}{\int f_a(x,Y_1,\dots,Y_n)\dee x},
	\]
	where we use $f_a(x,Y_1,\dots,Y_n)$ to denote the joint distribution of $x$ and $\{Y_i\}$ when $\IG(a,b)$ is used as the prior for ${\sigma^2_A}$. That is,
	\[
	\begin{split}
	&f_a(x,Y_1,\dots,Y_n)\\
	&=\frac{b^a}{\Gamma(a)}({\sigma^2_A})^{-a-1}e^{-b/{\sigma^2_A}}\prod_{i=1}^n\frac{1}{\sqrt{2\pi {\sigma^2_A}}}e^{-\frac{(\theta_i-\mu)^2}{2{\sigma^2_A}}}\frac{1}{\sqrt{2\pi}}e^{-\frac{(Y_i-\theta_i)^2}{2{\sigma_V^2}}}\\
	&=\frac{1}{(2\pi)^n}\frac{b^a}{\Gamma(a)}({\sigma^2_A})^{-a-1-\frac{n}{2}}e^{-b/{\sigma^2_A}}\exp\left[-\sum_{i=1}^n\left(\frac{(\theta_i-\mu)^2}{2{\sigma^2_A}}+\frac{(Y_i-\theta_i)^2}{2{\sigma_V^2}}\right)\right].
	\end{split}
	\]
	Now using $\frac{1}{{\sigma^2_A}}f_a(x, Y_1,\dots,Y_n)=\frac{a}{b}f_{a+1}(x,Y_1,\dots,Y_n)$, we have
	\[
	\EE_{\pi}(1/{\sigma^2_A})=\frac{a}{b}\frac{\int f_{a+1}(x,Y_1,\dots,Y_n)\dee x}{\int f_a(x,Y_1,\dots,Y_n)\dee x},\quad 	\EE_{\pi}(1/({\sigma^2_A})^2)=\frac{a^2}{b^2}\frac{\int f_{a+2}(x,Y_1,\dots,Y_n)\dee x}{\int f_a(x,Y_1,\dots,Y_n)\dee x}.
	\]
	Therefore, it suffices to show the ratios $\frac{\int f_{a+1}(x,Y_1,\dots,Y_n)\dee x}{\int f_a(x,Y_1,\dots,Y_n)\dee x}$ and $\frac{\int f_{a+2}(x,Y_1,\dots,Y_n)\dee x}{\int f_a(x,Y_1,\dots,Y_n)\dee x}$ are (asymptotically) bounded. Next, we focus on the first ratio. The second ratio can be proved using a similar argument.
	
	Using the fact that
	\[
	\begin{split}
	&\int\exp\left[-\left(\frac{{\sigma_V^2}(\theta_i-\mu)^2+{\sigma^2_A}(Y_i-\theta_i)^2}{2{\sigma^2_A}{\sigma_V^2}}\right)\right]\dee \theta_i\\
	&=\left(\int \exp\left[-\frac{\left(\theta-\frac{{\sigma_V^2}\mu+Y {\sigma^2_A}}{{\sigma^2_A}+{\sigma_V^2}}\right)^2}{\frac{2{\sigma^2_A} {\sigma_V^2}}{{\sigma^2_A}+{\sigma_V^2}}}\right]\dee \theta_i\right) \left(\exp\left[-\frac{(Y_i-\mu)^2}{2({\sigma_V^2}+{\sigma^2_A})}\right]\right)\\
	&=\sqrt{2\pi \frac{2{\sigma^2_A}{\sigma_V^2}}{{\sigma_V^2}+{\sigma^2_A}}}\exp\left[-\frac{(Y_i-\mu)^2}{2({\sigma_V^2}+{\sigma^2_A})}\right],
	\end{split}
	\]
	and
	\[
	\begin{split}
	&\int\exp\left[-\frac{\sum_{i=1}^n(Y_i-\mu)^2}{2({\sigma_V^2}+{\sigma^2_A})}\right]\dee\mu\\
	&=\left(\int \exp\left[-\frac{(\mu-\bar{Y})^2}{2({\sigma_V^2}+{\sigma^2_A})/n}\right]\dee\mu\right) \left(\exp\left[-\frac{\sum_i Y_i^2-n \bar{Y}^2}{2({\sigma_V^2}+{\sigma^2_A})}\right]\right)\\
	&=\exp\left[-\frac{\sum_{i=1}^n(Y_i-\bar{Y})^2}{2({\sigma_V^2}+{\sigma^2_A})}\right]\sqrt{2\pi \frac{2({\sigma_V^2}+{\sigma^2_A})}{n}},
	\end{split}
	\]
	we can write $\EE_{\pi}(1/{\sigma^2_A})$ as a function of $\Delta=\sum_i(Y_i-\bar{Y})^2$. Denote $h_n(\Delta):=\EE_{\pi}(1/{\sigma^2_A})$, then we have
	\[
	h_n(\Delta):=\frac{\int ({\sigma^2_A})^{-a-2}e^{-b/{\sigma^2_A}}({\sigma_V^2}+{\sigma^2_A})^{\frac{-n+1}{2}}\exp\left[-\frac{\Delta}{2({\sigma_V^2}+{\sigma^2_A})}\right]\dee {\sigma^2_A}}{\int ({\sigma^2_A})^{-a-1}e^{-b/{\sigma^2_A}}({\sigma_V^2}+{\sigma^2_A})^{\frac{-n+1}{2}}\exp\left[-\frac{\Delta}{2({\sigma_V^2}+{\sigma^2_A})}\right] \dee {\sigma^2_A}}.
	\]
	Next, we show $h_n((n-1)(c+{\sigma_V^2}))$ is (asymptotically) bounded for any fixed $c>0$. Note that 
	\[
	\begin{split}
	\int &({\sigma^2_A})^{-a-1}e^{-b/{\sigma^2_A}}({\sigma_V^2}+{\sigma^2_A})^{\frac{-n+1}{2}}\exp\left[-\frac{\Delta}{2({\sigma_V^2}+{\sigma^2_A})}\right]\dee {\sigma^2_A}\\
	&=\int ({\sigma^2_A})^{-a-1}e^{-b/{\sigma^2_A}}  \left\{\frac{1}{\sqrt{{\sigma_V^2}+{\sigma^2_A}}}\exp\left[-\frac{\frac{\Delta}{n-1}}{2({\sigma_V^2}+{\sigma^2_A})}\right]\right\}^{n-1} \dee {\sigma^2_A}.
	\end{split}
	\]
	We change variable $y=\frac{1}{\sqrt{{\sigma_V^2}+{\sigma^2_A}}}$ and apply the Laplace approximation. Note that for any $c>0$, let $y_0=\arg\max_{y}\left[y\exp\left(-\frac{c+{\sigma_V^2}}{2}y^2\right)\right]$, then $y_0=\frac{1}{\sqrt{c+{\sigma_V^2}}}$. Therefore, by the Laplace approximation \cite[Thm. 1, Chp. 19.2.4]{Zorich2004}, we have
	\[
	\begin{split}
	h_n((n-1)(c+{\sigma_V^2}))
	&=\frac{c^{-a-2}e^{-b/c}\left[y_0\exp\left(-\frac{c+{\sigma_V^2}}{2}y_0^2\right)\right]^{n-1}(1+\bigO(n^{-\frac{1}{2}}))}{c^{-a-1}e^{-b/c}\left[y_0\exp\left(-\frac{c+{\sigma_V^2}}{2}y_0^2\right)\right]^{n-1}(1+\bigO(n^{-\frac{1}{2}}))}\\
	&=\frac{1}{c}(1+\bigO(n^{-1/2})),
	\end{split}
	\]
	where the term $\bigO((n^{-1/2})$ only depends on constants $a$, $b$, and ${\sigma_V^2}$.
	Finally, since for all $n\ge N_0$ we have $\Delta\ge (n-1)({\sigma_V^2}+\delta)$, this implies
	$h_n(\Delta)\le\frac{1}{\delta}(1+\bigO(n^{-1/2})), \forall n\ge N_0$. Therefore, there exists large enough positive integer $N_0$, which only depends on $a$, $b$, ${\sigma_V^2}$, and $\delta$, such that for all $n\ge N_0$, we have $\EE_{\pi}(1/{\sigma^2_A})=h_n(\Delta)\le \frac{1}{\delta}(1+\bigO(n^{-1/2}))\le \frac{2}{\delta}$. 
	
	For $\EE_{\pi}(1/({\sigma^2_A})^2)$, we can follow a similar argument to show that $\EE_{\pi}(1/({\sigma^2_A})^2)\le \frac{2}{\delta^2}$ for large enough $n$. Therefore, we can conclude that there exists large enough positive integer $N$, which only depends on $a$, $b$, ${\sigma_V^2}$, and $\delta$, such that for all $n\ge N$, we have both $\EE_{\pi}(1/{\sigma^2_A})\le \frac{2}{\delta}$ and $\EE_{\pi}(1/({\sigma^2_A})^2)\le \frac{2}{\delta^2}$.
\end{proof}

	\section{Proof of Theorem \ref{thm_toy_example}}\label{proof_thm_toy_example}
	
The key step is to establish the following drift condition:
\[
\EE[f^{\textrm{new}}(X^{(1)})]\le \frac{1}{4} f^{\textrm{new}}(X^{(0)})+\bigO(1/p).
\]
It suffices to show 
\[
\EE\left[\left(\frac{\|X_1^{(1)}\|^2}{p}-1\right)^2\mid X_2^{(0)}\right]\le \lambda \left(\frac{\|X_2^{(0)}\|^2}{p}-1\right)^2+b,
\]
where $\lambda=\frac{1}{2}$ and $b=\bigO(1/p)$. 

Writing $X_2^{(0)}=x$ and $X_1^{(1)}=\frac{1}{2}x+Z$ where $Z\sim \mathcal{N}(0,\frac{3}{4}I_p)$, we have
\[
&\EE\left[\left(\frac{\|X_1^{(1)}\|^2}{p}-1\right)^2\mid X_2^{(0)}=x\right]\\
&=\left(\frac{\|x\|^2}{4p}-1\right)^2+2\left(\frac{\|x\|^2}{4p}-1\right)\EE\left[\frac{\|Z\|^2}{p}\right]+\EE\left[\left(\frac{\|Z\|^2+Z^Tx}{p}\right)^2\right]\\
&=\left(\frac{\|x\|^2}{4p}-1\right)^2+2\left(\frac{\|x\|^2}{4p}-1\right)\frac{3}{4} +\left(\frac{9}{16}+\frac{9}{8p}+\frac{3}{4p}\frac{\|x\|^2}{p}\right)\\
&=\frac{1}{4}\left(\frac{\|x\|^2}{p}-1+\frac{6}{p}\right)^2+\bigO(\frac{1}{p})\\
&=\frac{1}{2}\left(\frac{\|x\|^2}{p}-1\right)^2+\bigO(\frac{1}{p}),
\]
where the last step is by
\[
\left(\frac{\|x\|^2}{p}-1+\frac{6}{p}\right)^2\le 2\left[\left(\frac{\|x\|^2}{p}-1\right)^2+\frac{36}{p^2}\right].
\]

To complete the proof, we still need to show a \emph{multi-step} minorization condition with $\epsilon$ bounded away from zero. Note that the $1$-step drift condition directly implies a $k$-step drift condition with $\lambda=\frac{1}{2^k}$ and $b=\bigO(1/p)$. Next, note that
\[
X_2^{(k)}\mid X_2^{(0)}=x_2\sim \mathcal{N}(\frac{1}{4^{k}}x_2,(1-\frac{1}{4^{k+1}})I_p).
\]
Therefore, according to the $k$-step drift condition, for all the states $x$ in the small set, we have $c\sqrt{p}\le \|x_2\|\le C\sqrt{p}$ for some positive constant $c<1$ and $C>1$. Then we choose $k$ such that $\|x_2\|/4^k=\bigO(1/p)$ so that the integral of the minimum of the two one-dimensional densities $\mathcal{N}(\frac{1}{4^{k}}C\sqrt{p},(1-\frac{1}{4^{k+1}}))$ and $\mathcal{N}(-\frac{1}{4^{k}}C\sqrt{p},(1-\frac{1}{4^{k+1}}))$ is $1-\bigO(1/p)$. Then by writing the multivariate Gaussian density as product of one-dimensional densities, the total minimization volume can be controlled so that $\epsilon=(1-\bigO(1/p))^p>0$ and bounded away from zero as $p\to \infty$. Therefore, we can choose $k=\lfloor C\log(p)\rfloor+1$ a large enough constant $C$. Overall, we have proven that for a $k$-step drift condition and the corresponding minimization condition gives $\epsilon$ which is asymptotically bounded away from zero, which completes the proof.
	
	\section{Proof of Theorem \ref{thm_new_example}}\label{proof_thm_new_example}
	We analyze this model by choosing a drift function
\[
f_n(x)=\left(\frac{\bar{\lambda}}{\alpha}-\frac{1}{\beta}\right)^2
\]
where $\bar{\lambda}=\frac{1}{n}\sum_i\lambda_i$. The key step of the proof is to show the following drift condition
\[
\EE[f_n(X^{(k+1)})\mid x^{(k)}]\le b,
\]
where $b=\bigO(1/n)$. 
For simplicity of notation, we omit the index $k$ in the rest of the proof. The computation of $\EE[f_n(X^{(k+1)})\mid x^{(k)}]$ have two steps. We first compute the conditional expectation over $\beta\mid \lambda\sim \textrm{Ga}(\rho+n\alpha, \delta+n\bar{\lambda})$. Using the fact that $1/\beta$ has an inverse gamma distribution, we have
\[
\EE_{\beta\mid \lambda}\left[\left(\frac{\bar{\lambda}}{\alpha}-\frac{1}{\beta}\right)^2\right]=\left(\frac{\bar{\lambda}}{\alpha}\right)^2-2\left(\frac{\bar{\lambda}}{\alpha}\right)\left(\frac{\frac{\delta}{n}+\bar{\lambda}}{\frac{\rho-1}{n}+\alpha}\right)
+\frac{\left(\frac{\delta}{n}+\bar{\lambda}\right)^2}{\left(\frac{\rho-1}{n}+\alpha\right)\left(\frac{\rho-2}{n}+\alpha\right)}
\]
Next, we compute the conditional expectation over $\lambda$ given $\beta$. Note that by summing (conditional) independent Gamma distribution we know
\[
n\bar{\lambda}\mid \beta \sim \textrm{Ga}(n(\bar{Y}+\alpha),1+\beta)
\]
which gives
\[
\EE_{\lambda\mid \beta}[\bar{\lambda}]=\frac{\bar{Y}+\alpha}{1+\beta},\quad \EE_{\lambda\mid \beta}[\bar{\lambda}^2]=\frac{(\bar{Y}+\alpha)(\bar{Y}+\alpha+\frac{1}{n})}{(1+\beta)^2}.
\]
Using the assumption on $\bar{Y}$ and the fact that $\frac{1}{1+\beta}\in (0,1]$, we have
\[
&\EE_{\lambda\mid \beta}\left[\left(\frac{\bar{\lambda}}{\alpha}\right)^2-2\left(\frac{\bar{\lambda}}{\alpha}\right)\left(\frac{\frac{\delta}{n}+\bar{\lambda}}{\frac{\rho-1}{n}+\alpha}\right)
+\frac{\left(\frac{\delta}{n}+\bar{\lambda}\right)^2}{\left(\frac{\rho-1}{n}+\alpha\right)\left(\frac{\rho-2}{n}+\alpha\right)}\right]\\
&=\EE_{\lambda\mid \beta}\left[\left(\frac{\bar{\lambda}}{\alpha}-\frac{\frac{\delta}{n}+\bar{\lambda}}{\frac{\rho-1}{n}+\alpha}\right)^2+\frac{\frac{1}{n}\left(\frac{\delta}{n}+\bar{\lambda}\right)^2}{\left(\frac{\rho-1}{n}+\alpha\right)^2\left(\frac{\rho-2}{n}+\alpha\right)}\right]= \bigO(\frac{1}{n^2}) + \bigO(\frac{1}{n}).
\]
Therefore, we established the drift condition $\EE[f_n(X^{(k+1)})\mid x^{(k)}]\le b$ where $b=\bigO(1/n)$.

Now the proof can be completed by verifying the Gibbs sampler satisfies 
the minorization condition: $P(x,\cdot)\ge \epsilon Q(\cdot)$  for all $x$ in the small set $\left\{\left|\bar{\lambda}-\frac{\alpha}{\beta}\right|=\bigO(1/\sqrt{n})\right\}$. We only need to show that $\epsilon$ is asymptotically bounded away from $0$ as $n\to\infty$. Note that the last step of updating $\beta$ in the Gibbs sampler doesn't depend on the previous state, it then suffices to derive the minorization condition for the step $n\bar{\lambda}\mid \beta \sim \textrm{Ga}(n(\bar{Y}+\alpha),1+\beta)$ for all $\beta$ in the small set. Let $\beta_{\max}$ and $\beta_{\min}$ be the maximum and minimum value of $\beta$ in the small set. Then from the explicit form of the density of $\bar{\lambda}$, on can see that $\epsilon$ must be asymptotically bounded away from $0$ if $1/(1+\beta_{\min})-1/(1+\beta_{\max})=\bigO(1/\sqrt{n})$, which is satisfied by the small set. This completes the proof.

	\section{Proof of Remark \ref{key_remark}}\label{proof_key_remark}
	
    \citet[Appendix C]{jones2006fixed} states another way to obtain samples from the posterior of the MCMC model related to James--Stein estimator. More specifically, recall that the model 
    \[
	\begin{split}
	Y_i~|~\theta_i\quad &\sim \Normal(\theta_i,\sigma^2_V),\quad 1\le i\le n,\\
	\theta_i~|~\mu,\sigma^2_A &\sim \Normal(\mu,\sigma^2_A), \quad 1\le i\le n,\\
	\mu &\sim \textrm{ flat prior on }\mathbb{R},\\
	\sigma^2_A &\sim \IG(a,b),
	\end{split}
	\]
	where {$\sigma^2_V$ is assumed to be known}, $Y=(Y_1,\dots,Y_n)$ is the observed data, and $x=(\sigma^2_A,\mu, \theta_1,\dots,\theta_n)$ are parameters. Then the posterior can be written as
	\[
	\pi(\theta, \mu, \sigma^2_A \mid Y)=\pi(\theta\mid \mu,\sigma^2_A,Y)\pi(\mu\mid \sigma^2_A, Y)\pi(\sigma^2_A\mid Y),
	\]
	where $\pi(\theta\mid \mu, \sigma^2_A,Y)$ is a product of independent univariate normal densities 
	\[
	\theta_i\sim \mathcal{N}\left(\frac{\sigma^2_A Y_i+\sigma^2_V\mu}{\sigma^2_V+\sigma^2_A},\frac{\sigma^2_A\sigma^2_V}{\sigma^2_A+\sigma^2_V}\right)
	\]
	and $\pi(\mu\mid \sigma^2_A, Y)$ is a normal distribution
	\[
	\mu\mid \sigma^2_A,Y \sim \mathcal{N}\left(\bar{Y}, \frac{\sigma^2_A+\sigma^2_V}{n}\right)
	\]
	Therefore, one can use a rejection sampler with proposal from $\IG(a,b)$ to obtain independent samples from $\pi(\sigma^2_A\mid Y)$. However, we show that the acceptance probability of this rejection sampler decreases (typically exponentially) fast with $n$. To see this, note that
	\[
\pi(\sigma^2_A\mid Y)\propto \frac{1}{(\sigma^2_A)^{a+1}(\sigma^2_A+\sigma^2_V)^{(n-1)/2}}\exp(-\frac{1}{b}-\frac{\sum_{i=1}^n(Y_i-\bar{Y})^2}{2(\sigma^2_A+\sigma^2_V)}).
	\]
	We let $g(\sigma^2_A)$ be the density of $\IG(a,b)$, then using the fact
	\[
		\frac{\pi(\sigma^2_A\mid Y)}{g(\sigma^2_A)}&\propto\frac{1}{(\sigma^2_A)^{a+1}(\sigma^2_A+\sigma^2_V)^{(n-1)/2}}\exp(-\frac{1}{b}-\frac{\sum_{i=1}^n(Y_i-\bar{Y})^2}{2(\sigma^2_A+\sigma^2_V)})/g(\sigma^2_A)\\
	&=(\sigma^2_A+\sigma^2_V)^{(1-n)/2}\exp(-\sum_{i=1}^n(Y_i-\bar{Y})^2/2(\sigma^2_A+\sigma^2_V))\\
	&\le M:= \left(\frac{\sum_{i=1}^n(Y_i-\bar{Y})^2}{n-1}\right)^{(1-n)/2}e^{-\frac{n-1}{2}}
	\]
	where the upper bound $M$ is achieved when $\sigma^2_A=\frac{\sum_{i=1}^n(Y_i-\bar{Y})^2}{n-1}-\sigma^2_V$.
	Then the acceptance probability of the rejection sampler is
	\[
	&\EE_{ \sigma^2_A\sim \IG(a,b)}\left[\frac{(\sigma^2_A+\sigma^2_V)^{(1-n)/2}\exp(-\sum_{i=1}^n(Y_i-\bar{Y})^2/2(\sigma^2_A+\sigma^2_V))}{M}\right]\\
	&=\EE_{ \sigma^2_A\sim \IG(a,b)}\left[\left(\frac{\sigma^2_A+\sigma^2_V}{\frac{\sum_{i=1}^n(Y_i-\bar{Y})^2}{n-1}}\right)^{(1-n)/2}\exp\left(\frac{\frac{\sum_{i=1}^n(Y_i-\bar{Y})^2}{n-1}}{\sigma^2_A+\sigma^2_V}-1\right)^{(1-n)/2}\right]\\
		&=\EE_{ \sigma^2_A\sim \IG(a,b)}\left[\left(\frac{\sigma^2_A+\sigma^2_V}{\frac{\sum_{i=1}^n(Y_i-\bar{Y})^2}{n-1}}\right)^{(1-n)/2}\exp\left(\frac{\frac{\sum_{i=1}^n(Y_i-\bar{Y})^2}{n-1}}{\sigma^2_A+\sigma^2_V}-1\right)^{(1-n)/2}\right].
	\]
	Therefore, the acceptance probability of the rejection sampler equals to $\EE[Z^{(n-1)/2}]$ where
    \[
	Z:=\left(\frac{\sigma^2_A+\sigma^2_V}{\frac{\sum_{i=1}^n(Y_i-\bar{Y})^2}{n-1}}\right)^{-1}\exp\left(1-\frac{\frac{\sum_{i=1}^n(Y_i-\bar{Y})^2}{n-1}}{\sigma^2_A+\sigma^2_V}\right)\le 1,
	\]	
	where the last inequality comes from $\exp(x-1)\ge x$.
	
	We can see that under mild conditions such that $\frac{\sum_{i=1}^n(Y_i-\bar{Y})^2}{n-1}$ converges to a constant, the acceptance probability of the rejection sampler goes to zero, $\EE[Z^{(n-1)/2}]\to 0$, very fast.

\end{appendix}

\begin{acks}[Acknowledgments]
The authors thank Jim Hobert and Gareth Roberts for helpful discussions, {and two referees for their valuable comments which have significantly improved the quality of the paper}. J.Y.\ also thanks Quan Zhou and Aaron Smith for helpful comments on the proof. This research is supported by the Natural Sciences and Engineering Research Council (NSERC) of Canada.
\end{acks}
\bibliographystyle{imsart-number} 
\bibliography{accepted-version/mcmc_bib.bib}       


\end{document}